\documentclass{lmcs}
\pdfoutput=1

\usepackage{lastpage}
\lmcsdoi{15}{1}{7}
\lmcsheading{}{\pageref{LastPage}}{}{}%
{Jul.~19,~2017}{Jan.~31,~2019}{}

\usepackage{amssymb,amsmath,euscript,mathabx}
\usepackage{graphicx}
\usepackage{tikz}
\usetikzlibrary{arrows}
\tikzstyle{every picture} = [>=latex]

\usepackage{xspace}

\def\cf#1{{\EuScript#1}}
\let\sem\setminus
\def\TM#1#2{{\cf T\cf M}_{#2}(#1)}
\def\SC#1{{\cf S\cf C}(#1)}
\def\Clos{\mbox{\rm Clos}}
\def\TMC#1#2#3{{\cf T\cf M}_{#2}^{#3}(#1)}

\newcommand{\MSO}{\ensuremath{\mathrm{MSO}}\xspace}
\newcommand{\MSOi}{\ensuremath{\mathrm{MSO}_1}\xspace}
\newcommand{\CMSOi}{\ensuremath{\mathrm{CMSO}_1}\xspace}
\newcommand{\MSOii}{\ensuremath{\mathrm{MSO}_2}\xspace}

\def\prebox#1{\mathop{\mbox{\sl #1}}}


\begin{document}

\title{Shrub-depth: Capturing Height of Dense Graphs}

\author[R.~Ganian]{Robert Ganian}
\address{Algorithms and Complexity Group, TU Wien, Vienna, Austria}
 \email{rganian@gmail.com}
\author[P.~Hlin\v{e}n\'y]{Petr Hlin\v{e}n\'y}
\address{Faculty of Informatics, Masaryk University, Brno, Czech Republic}
 \email{hlineny@fi.muni.cz}
\author[J.~Ne{\v s}et{\v r}il]{Jaroslav Ne{\v s}et{\v r}il}
\address{Computer Science Inst.\ of Charles University (IUUK), Praha, Czech Republic}
 \email{nesetril@iuuk.mff.cuni.cz}
\author[J.~Obdr\v{z}\'alek]{Jan Obdr\v{z}\'alek}
\address{Faculty of Informatics, Masaryk University, Brno, Czech Republic}
 \email{obdrzalek@fi.muni.cz}
\author[P.~Ossona~de~Mendez]{Patrice~Ossona~de~Mendez}
\address{CNRS UMR 8557, \'Ecole des Hautes \'Etudes en Sciences Sociales, Paris, France}
 \email{pom@ehess.fr}

 \thanks{R.~Ganian, P.~Hlin\v{e}n\'y, J.~Ne{\v s}et{\v r}il and
	J.~Obdr\v{z}\'alek have been
	supported by the Institute for Theoretical Computer Science (CE-ITI),
	Czech Science Foundation project No. P202/12/G061.
 	J.~Ne{\v s}et{\v r}il and P.~Ossona~de~Mendez have been supported by the project
	LL1201 CORES of the Ministry of Education of the Czech republic.
    Robert Ganian also acknowledges support from the Austrian Science Fund (FWF, project P31336).
    }

\begin{abstract}
  The recent increase of interest in the graph invariant
  called {\em tree-depth} and in its applications in algorithms and logic on
  graphs led to a natural question:
  is there an analogously useful ``depth'' notion also for dense graphs
  (say; one which is stable under graph complementation)?
  To this end, in a 2012 conference paper,
  a new notion of {\em shrub-depth} has been introduced,
  such that it is related to the established notion of clique-width in a similar way as
  tree-depth is related to tree-width.
  Since then shrub-depth has been successfully used in several research papers.
  Here we provide an in-depth review of the definition and basic properties
  of shrub-depth, and we focus on its logical aspects which turned
  out to be most useful.
  In particular, we use shrub-depth to give a characterization of the lower
  $\omega$ levels of the MSO$_1$ transduction hierarchy of simple graphs.
%
\end{abstract}

\maketitle

\section{Introduction}
\label{sec:introduction}

In this paper, we are interested in a structural graph parameter that is intermediate
between clique-width and tree-depth, sharing the nice properties of
both. Clique-width, originated by Courcelle et al in~\cite{CourcelleER1991,co00}, is the older of the two notions. 
In several aspects, the theory of graphs of bounded clique-width is similar to
the one of bounded tree-width. Indeed, bounded tree-width implies bounded clique-width. 
However, unlike tree-width, graphs of bounded clique-width
include arbitrarily large cliques and other dense graphs,
and the value of clique-width does not change much when complementing the edge set of a graph.
Clique-width is not closed under taking subgraphs or minors, 
only under taking induced subgraphs. 
As we will see later, clique-width is also closely related to trees and 
monadic second-order logic of graphs.

The notion of tree-depth of a graph, coined by Ne{\v s}et{\v r}il
and Ossona~de~Mendez \cite{no06}, is equivalent or similar to some older notions 
such as the {\em vertex ranking number} and the minimum height of an 
{\em elimination tree} \cite{bodlaenderetal98,dkkm94,sch90}, etc. 
Graphs of small tree-depth are related to trees of small height, 
and they enjoy strong ``finiteness'' properties
(finiteness of cores, existence of non-trivial automorphisms if the graph is
large, well-quasi-ordering by subgraph inclusion).
The tree-depth notion received almost immediate attention, as it plays 
a central role in the theory of graph classes of bounded expansion \cite{no08,sparsity}.
However, graphs of small tree-depth are necessarily very sparse and the
notion behaves badly with respect to, say, graph complementation.

\smallskip

Our search for a structural concept ``between clique-width and tree-depth''
\cite{ghnoor12} has originally been inspired by algorithmic considerations:
graphs of bounded parameters such as clique-width allow efficient
solvability of various problems which are difficult (e.g. NP-hard)
in general, e.g.~\cite{cmr00,egw01,gk03,gho11}.
Highly regarded results in this area are those which, instead of solving one
problem, give a solution to a whole class of problems (called
\emph{algorithmic metatheorems}). 
The perhaps most famous result of this kind is Courcelle's theorem~\cite{cou90}, 
which states that every graph property expressible in the \MSOii logic of graphs 
can be solved in time ${\cf O}(|G|.f(\phi, k))$ where
$f$ is a computable function, meaning that the problem is
\emph{fixed-parameter tractable} 
(FPT for short). 
For clique-width, a result similar to Courcelle's theorem holds; \MSOi model
checking is FPT on graphs parameterized by clique-width~\cite{cmr00}.

However, an issue with these results is that, as showed by Frick and
Grohe~\cite{fg04} for \MSO model checking of the class of all trees, the
function $f$ of Courcelle's algorithm is, unavoidably,
non-elementary in the parameter $\phi$ (unless P=NP). 
This brings the following question: are there interesting graph classes
in which the runtime dependency on the formula $\phi$ is better?  
For instance, in 2010, Lampis~\cite{lam12} gave an FPT algorithm for \MSOii model checking on graphs
of bounded vertex cover with elementary (doubly-exponential) dependence on the formula.
Subsequently, in 2012, Gajarsk\'y and Hlin\v{e}n\'y showed~\cite{gh15}
that there exists a linear-time FPT algorithm for \MSOii model checking
of graphs of bounded tree-depth, again with elementary dependence on the formula. 
Their result is essentially best possible, as shown soon after by Lampis~\cite{lam14}.
In order to extend that result towards \MSOi model checking of (some classes of)
dense graphs, one would first need to adjust the clique-width concept towards 
``bounded depth'' (as with tree-depth), which is not a simple task.

The aforementioned paper \cite{gh15} was not the first one explicitly
raising the issue of restricting clique-width towards bounded depth in the literature.
In 2012, for example, independently Elberfeld, Grohe and Tantau made the following 
remark regarding the expressive power of graph FO logic~\cite{EGT12}:
{\it One idea is to develop an adjusted notion of clique-width that has the 
same relation to clique-width as tree-depth has to tree-width.}
Our concept of {shrub-depth}~\cite{ghnoor12} has provided a
quick positive answer also to the question of~\cite{EGT12}.
Clique-width-like graph decompositions of limited depth have also been used
as a tool by Blumensath and Courcelle in~\cite{bc14} (under the name
``$\otimes$-decompositions'').
However, some of their technical results which may be interesting 
in our context have not been published anywhere.

\smallskip

In \cite{ghnoor12}, two new structural depth parameters of graphs have been
introduced: \emph{shrub-depth} (Definition~\ref{def:shrub-depth}) and SC-depth (Definition~\ref{def:SC-depth}), which are asymptotically equivalent to each other.
Since their emergence these have been successfully used in several
research papers, and shrub-depth in particular 
is a subject of ongoing interest in the finite model theory of graphs.

For instance, the aforementioned~\cite{gh15} (its full journal version, to be
precise) has also extended \MSOii model checking tractability on graphs of bounded 
tree-depth to \MSOi on graph classes of bounded shrub-depth, again with an 
elementary runtime dependence on the checked formula.
Furthermore, \cite{gh15} has generalized the result of~\cite{EGT12}
to prove that the expressive power of FO and MSO$_1$ is the same on classes
of bounded shrub-depth.

In a recent paper by Gajarsk{\'{y}}, Kreutzer, Ne\v{s}et\v{r}il, Ossona de
Mendez, Pilipczuk, Siebertz and Toru{'{n}}czyk~\cite{DBLP:conf/icalp/GajarskyKNMPST18}, the
concept of shrub-depth has been successfully used to obtain an analog of low
tree-depth decompositions for transductions of bounded expansion classes.

On another topic, Hlin\v{e}n{\'{y}}, Kwon, Obdr\v{z}{\'{a}}lek and
Ordyniak~\cite{hkoo16} have shown that the tree-depth and shrub-depth
concepts of graphs are tightly related to each other via the so called vertex-minors.
Regarding alternative and generalized views of shrub-depth,
DeVos, Kwon and Oum [unpublished] in an ongoing work elaborate on the concept of
branch-depth of matroids, and prove that a derived new concept of 
rank-depth of graphs is asymptotically equivalent to shrub-depth.

\paragraph{Paper organization. }
Since the core initial paper on shrub-depth~\cite{ghnoor12} has appeared only as a
short conference version, we take an opportunity here to give a detailed
review of this concept and to provide full proofs of the results
of~\cite{ghnoor12} enhanced in light of the current state-of-the-art.
After preliminary definitions in Section~\ref{sec:definitions}, this overview
of shrub-depth and its structural properties 
(such as Theorems~\ref{thm:shrubsc}, \ref{thm:Path-TM} and~\ref{thm:clforbidden})
constitute Section~\ref{sec:shrub} of this paper.
The subsequent Section~\ref{sec:transdhier} focuses on logical aspects of
shrub-depth, which have so far been of greatest interest, and presents our
main results with their proofs.
We start with proving that the concept of shrub-depth of a graph class 
is stable -- meaning that the shrub-depth value does not grow,
under MSO$_1$ interpretations (Theorem~\ref{thm:shrub-depth-interpretability})
and also under non-copying MSO$_1$ transductions (Theorem~\ref{thm:nc-transduction}).
From that we derive (Theorem~\ref{thm:hierarchy}) that
the integer values of shrub-depth define the lower $\omega$ levels 
of the MSO$_1$ transduction hierarchy of simple graphs, 
which partially answers an open question raised by
Blumensath and Courcelle in~\cite{bc10}.
We conclude with some remarks and open questions in Section~\ref{sec:conclu}.

\section{Common Definitions}
\label{sec:definitions}

We assume the reader is familiar with the standard notation of graph theory. 
In particular, our graphs are finite, undirected and
simple (i.e. without loops or multiple edges). For a graph $G=(V,E)$ we use
$V(G)$ to denote its vertex set and $E(G)$ to denote the set of its edges. 
We write $G\simeq H$ to say that graphs $G$ and $H$ are isomorphic,
and similarly we use $G\subseteq H$ to say that $G$ is a subgraph of~$H$
(not necessarily induced). 
An isomorphism of a graph to itself is also called an {\em automorphism}.
We will also use \emph{labelled graphs}, where each vertex is assigned one or more
of a fixed finite set of labels (in this case, isomorphism implicitly
preserves the labels).

A forest $F$ is a graph without cycles, and a
tree $T$ is a forest with a single connected component. We will consider
mainly \emph{rooted forests (trees)}, in which every connected component has
a designated vertex called the \emph{root}.  The {\em height} of a vertex
$x$ in a rooted forest $F$ is the length of a path from the root (of the
component of $F$ to which $x$ belongs) to $x$.  The {\em height}\,\footnote{ There is a conflict in the
  literature about whether the height of a rooted tree should be measured by
  the ``root-to-leaves distance'' or by the ``number of levels'' (a
  difference of $1$ on finite trees). We adopt the convention that the
  height of a single-node tree is $0$ (i.e., the former view).} of the
rooted forest $F$ is the maximum height of the vertices of $F$. Let $x,y$ be
vertices of $F$. The vertex $x$ is an {\em ancestor} of $y$, and $y$ is a
{\em descendant} of $x$, in $F$ if $x$ belongs to the path of $F$ linking
$y$ to the corresponding root; we denote this as $y\leq x$ in F.
If $x$ is an ancestor of $y$ and $xy\in E(T)$, 
then $x$ is called a {\em parent} of $y$, and $y$ is a {\em child} of~$x$.
The least common ancestor of $x$ and $z$ in $F$ is denoted by $x\wedge z$.

\subsection{Width and depth measures}
\label{sub:widths}

The so called width measures play an important role in structural graph
theory and in its algorithmic applications.
A prototypical width parameter is the {\em tree-width} of a graph \cite{gmii}
introduced by Robertson and Seymour together with the related {\em path-width}.
We refer to \cite{die05} for missing definitions and basic properties.

The primary interest of our paper is in two other, seemingly unrelated, 
structural width measures which we define now.

\begin{defi}[Clique-width \cite{CourcelleER1991,co00}]
\label{def:cw}
A \emph{$k$-expression} is an algebraic expression
having the following four operations on vertex-labelled graphs using $k$
labels:
\begin{itemize}
\item create a new vertex with a single label $i$;
\item take the disjoint union of two labelled graphs;
\item add all edges between vertices of label $i$ and label $j$ ($i\not=j$); and
\item relabel all vertices with label $i$ to label $j$.
\end{itemize}
The {\em clique-width} ${\rm cw}(G)$ of a graph $G$ equals the minimum $k$ such that
(some labelling of) $G$ is the value of a $k$-expression.
\end{defi}

Clique-width may be low even on graph classes for which the tree-width is
unbounded, such as complete graphs or complete bipartite graphs (the 
clique-width of which is 2).
Note that Definition~\ref{def:cw} demands each vertex to carry only one
label, while one can allow multiple labels as well.
Another possible modification is to allow $i=j$ in the third step.
Both these relaxations, while changing values of clique-width for some particular
graphs, are nevertheless asymptotically equivalent to the standard 
clique-width notion of Definition~\ref{def:cw}.

One can, furthermore, define {\em linear clique-width} (see, e.g., \cite{hmp12}) which has the additional
restriction that the union operator is allowed to take only a single vertex
as the right-hand operand (i.e., the expression tree is a caterpillar---this
is conceptually related to path-width).

A close alternative of clique-width is represented by the NLC classes
introduced by Wanke~\cite{Wan94}.
${\rm NLC}_m$ consists of all graphs that can be obtained
from single vertices with single labels in $\{1,\dots,m\}$ 
using the two following operations:
\begin{itemize}
  \item disjoint union of two graphs $G_1$ and $G_2$, with addition of all edges
  between vertices of $G_1$ with label $i$ and vertices of $G_2$ with label $j$
  whenever $(i,j)$ belongs to a given fixed subset $S$ of
  $\{1,\dots,m\}\times\{1,\dots,m\}$;
  \item relabelling of the vertices according to some map 
	$\{1,\dots,m\}\to\{1,\dots,m\}$.
\end{itemize}
The {\em NLC-width} of a graph is the minimum $m$
such that the graph belongs to ${\rm NLC}_m$. It has been proved in
\cite{joh98} that the NLC-width and the clique-width (cw) of a graph $G$ are
related by $\text{NLC-width}(G)\leq\text{cw}(G)\leq 2\cdot\text{NLC-width}(G)$.

At last, we briefly mention that another graph measure asymptotically equivalent 
to clique-width is {\em rank-width}~\cite{OS06}.
Similarly, linear clique-width is asymptotically equivalent 
to linear rank-width \cite{gan10}.

\smallskip
The second structural measure of our interest is tree-depth.

\begin{defi}[Tree-depth \cite{no06}]
\label{def:tree-depth}
The {\em closure} ${\Clos}(F)$ of a forest $F$
is the graph obtained from $F$ by making every vertex adjacent to all of its
ancestors.
The {\em tree-depth} ${\rm td}(G)$ of a graph $G$ is one more than the minimum
height of a rooted forest $F$ such that $G\subseteq{\Clos}(F)$.
\end{defi}
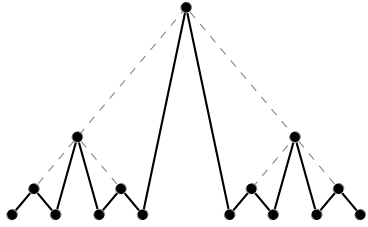
\begin{figure}[t]
{\hfill
\begin{tikzpicture}[scale=.9]
\tikzstyle{every node}=[draw, shape=circle, minimum size=4pt,inner sep=0pt, fill=black]
    \draw (0,6) node (a) {}
        -- ++(230:2.5cm) [color=gray,dashed] node (b) {}
        -- ++(230:1.0cm) node (c) {}
        -- ++(230:0.5cm) node (d) {};
    \draw  (a) {}
        -- ++(310:2.5cm) [color=gray,dashed] node (e) {}
        -- ++(230:1.0cm) node (f) {}
        -- ++(230:0.5cm) node (g) {};
    \draw  (b) {}
        -- ++(310:1.0cm) [color=gray,dashed] node (h) {}
        -- ++(230:0.5cm) node (i) {};
    \draw  (e) {}
        -- ++(310:1.0cm) [color=gray,dashed] node (j) {}
        -- ++(230:0.5cm) node (k) {};
    \draw  (c) {}
        -- ++(310:0.5cm) [color=gray,dashed] node (l) {} ;
    \draw  (f) {}
        -- ++(310:0.5cm) [color=gray,dashed] node (m) {} ;
    \draw  (h) {}
        -- ++(310:0.5cm) [color=gray,dashed] node (n) {} ;
    \draw  (j) {}
        -- ++(310:0.5cm) [color=gray,dashed] node (o) {} ;
    \draw (d) -- (c) -- (l) -- (b) -- (i) -- (h) -- (n) -- (a) -- (g)
	-- (f) -- (m) -- (e) -- (k) -- (j) -- (o) [thick];
\end{tikzpicture}
\hfill}
\caption{The path of length $n$ has tree-depth $\log_2(n+2)$,
	as in the depicted decomposition.}
\label{fig:td-path}
\end{figure}

Definition~\ref{def:tree-depth} is illustrated in Figure~\ref{fig:td-path}.
For a proof of the following proposition, as well as for a more
extensive study of tree-depth, we refer the reader to~\cite{sparsity}. 

\begin{prop}
\label{prop:td-basics}
  Let $G$ and $H$ be graphs. Then the following are true:
  \begin{enumerate}[label={\it\alph*)}]
  \item If $H$ is a minor of $G$, then ${\rm td}(H)\leq{\rm td}(G)$.
  \item If $L$ is the length of a longest path in $G$, then
  $\lceil \log_2 (L+2)\rceil\leq {\rm td}(G)\leq L+1$.
  \item If ${\rm tw}(G)$ and ${\rm pw}(G)$ denote the tree-width
  and path-width of a graph~$G$, then
  ${\rm tw}(G)\leq{\rm pw}(G)\leq {\rm td}(G)-1$.
  \end{enumerate}
\end{prop}

\subsection{\MSO logic on graphs}

We now briefly introduce \emph{monadic second order logic (\MSO)} over 
graphs and the concepts of MSO interpretation and transduction. 
We refer interested readers to, e.g., Courcelle and Engelfriet~\cite{ce12}
for further reading. In general,
\MSO is the extension of first-order logic by quantification over sets.
In our paper we deal with the following particular flavour:
\begin{defi}[\MSOi and \CMSOi logic of graphs]
\label{def:CMSO1graphs}
  The language of \MSOi consists of expressions built from the following
  elements:
  \begin{itemize}
  \item variables $x,y,\ldots$ for vertices, and $X,Y$ for sets of vertices,
  \item equality for variables, quantifiers $\forall,\exists$ ranging over
   vertices and vertex sets, and the standard Boolean connectives,
  \item the predicates $x\in X$ and $\prebox{edge}(x,y)$ with their standard
   meaning.
  \end{itemize}
One may also use an arbitrary number of unary predicates
on the vertex set (as vertex labels).
The language of \CMSOi ({\em counting \MSOi}), moreover, 
adds the predicates $\!\mod\!_{a,b}$, such that
$\!\mod\!_{a,b}(X)$ holds true if and only if $|X|\!\!\mod b=a$.
\end{defi}

\MSOi logic can be used to express many interesting graph
properties, such as 3-colourability and dominating set. 
We also briefly mention \MSOii logic of graphs, which additionally includes 
quantification over edge sets and can express properties which are not 
definable in \MSOi (e.g., Hamiltonicity).

From an algorithmic perspective, \MSO logic is particularly useful as the
language for describing tractable problems in algorithmic metatheorems
(e.g., for the aforementioned graphs of bounded clique-width
\cite{cmr00} or tree-width \cite{cou90}).
In this respect, we consider the $\cf L$-{\em model checking} problem in which the
input is a graph~$G$, the parameter is a formula $\phi$ of the considered logic $\cf L$ (such as \MSOi), and the question is whether $G\models\phi$.

\smallskip
A powerful tool, both in theory and in algorithmic metatheorems,
is the ability to ``efficiently translate'' an instance of the model
checking problem over a given class, into an instance of the problem over another class
(for which we, perhaps, already have an efficient model checking algorithm). 
We start with simple interpretations of undirected graphs.

\begin{defi}
\label{def:interpretation}
  A {\em simple \MSOi graph interpretation} is a pair $I=(\nu,\mu)$ of
  \MSOi formulae (with $1$ and $2$ free first order variables, respectively),
  such that $\mu$ is symmetric
  (i.e.,~$G\models\mu(x,y)\leftrightarrow\mu(y,x)$ in every graph $G$).%
\footnote{We remark that while the question whether $\mu$ is symmetric
is generally undecidable, we may simply force it to be symmetric, e.g., by
using $\mu(x,y)\vee\mu(y,x)$.
}
  To each graph $G$, the interpretation $I$ associates a graph $I(G)$ 
  which is defined as follows:
\begin{itemize}
  \item The vertex set of $I(G)$ (the {\em domain} of~$I$ in~$G$)
  is the set of all vertices $v$ of $G$ such that $G\models \nu(v)$;
  \item the edge set of $I(G)$ is the set of all the pairs $\{u,v\}$ of
  vertices of $G$ such that $G\models \nu(u)\wedge\nu(v)\wedge\mu(u,v)$.
\end{itemize}
A simple \CMSOi graph interpretation is defined analogously.
\end{defi}

For example, a complete graph can be interpreted in any graph (with the same
number of vertices) by letting $\nu\equiv\mu\equiv true$,
and the complement of a graph has an interpretation using $\nu\equiv true$ and
$\mu(x,y)\equiv\neg\prebox{edge}(x,y)$.

Note that, to each \CMSOi formula $\phi$, an interpretation
$I=(\nu,\mu)$ naturally and efficiently assigns a formula $I(\phi)$
such that $G\models I(\phi) \iff I(G)\models \phi$ holds.
Having classes $\cf G,\cf H$ of finite graphs, we say that $I$ is a
{\em simple interpretation of $\cf G$ in $\cf H$}
if the following holds:
for every $G\in\cf G$ there is $H\in\cf H$ such that $I(H)\simeq G$, and
for every $H\in\cf H$ it holds that $I(H)\in\cf G$.

\smallskip

A more general concept of a ``logical translation'' is that of transductions.
Briefly saying, in an addition to a simple interpretation, this allows
to add to a graph arbitrary ``parameters'' (as unary predicates) and to make
several disjoint copies of the graph.
A thorough discussion of this concept can be found in~\cite{ce12},
but we prefer to keep the paper simple and accessible to 
a wide audience of graph theorists,
and so we give a simplified version of the definition from~\cite{bc10}.

Still, before proceeding to Definition~\ref{def:transduction}, we have to
briefly extend the notion of interpretation towards 
finite relational structures with finite signatures.
A {\em relational structure} $\cf S=(U,\>R_1^{\cf S},\dots,R_a^{\cf S})$
of the signature $\sigma=\{R_1,\dots,R_a\}$ consists of a
universe (a finite set) $U$ and a (finite) list of relations
$R_1^{\cf S},\dots,R_a^{\cf S}$ over~$U$.
For instance, for graphs, $U=V(G)$ is the vertex set and $R_1^{G}=E(G)$ 
is the binary symmetric relation of edges of~$G$.
The language of \CMSOi logic 
of relational structures of the signature $\{R_1,\dots,R_a\}$
is as in Definition~\ref{def:CMSO1graphs} with the predicates
$R_1,\dots,R_a$ (instead of $\prebox{edge}$).
The scope of Definition~\ref{def:interpretation} 
of a simple graph interpretation $I=(\nu,\mu)$ is then naturally 
generalized by allowing $\nu$ and $\mu$ to be \CMSOi formulae
over relational structures of the signature $\sigma=\{R_1,\dots,R_a\}$.
For each structure $\cf S$ of the signature $\sigma$,
the interpretation $I(\cf S)$ is, in our case, a simple graph (again
possibly with arbitrarily assigned vertex labels).

\begin{defi}[\MSOi and \CMSOi transduction]
\label{def:transduction}
A {\em basic \MSOi graph transduction} $\tau_1$ is a triple $(\chi,\nu,\mu)$ such that
$(\nu,\mu)=I$ is a simple \MSOi graph interpretation, and $\chi$ is an \MSOi sentence.
The transduction $\tau_1$ maps a relational structure $\cf S$ to a graph $I(\cf S)$,
denoted here by $\tau_1(\cf S)$, if $\cf S\models\chi$,
and $\tau_1(\cf S)$ is undefined if $\cf S\not\models\chi$.

The $k$-copy operation maps a graph $G$ to the relational structure $G^{\times k}$ such that
$V(G^{\times k})=V(G)\times\{1,\dots,k\}$, the subset $V(G)\times\{i\}$ for
$i=1,2\dots,k$ induces a copy of the graph~$G$ (there are no edges between distinct
copies), and $V(G^{\times k})$ is additionally equipped with a binary relation
$\sim$ and unary $P_1,\dots,P_k$ such that;
$(u,i)\sim(v,j)$ for $u,v\in V(G)$ iff $u=v$, and $P_i=\{(v,i): v\in V(G)\}$.

The expansion of a graph $G$ by $p$ unary predicates maps $G$ to the set of all
structures obtained by an expansion of $V(G)$ by $p$ new unary predicates
(as vertex labels).

Altogether, a many-valued map $\tau$ is an {\em \MSOi graph transduction}
if it can be written as $\tau=\tau_1\circ\gamma\circ\varepsilon$, 
where $\tau_1$ is a basic graph transduction,
$\gamma$ is a $k$-copy operation for some $k$,
and $\varepsilon$ is the expansion by $p$ unary predicates for some~$p$.
Specially, if $k=1$, then we call $\tau$ a {\em non-copying transduction}.

A \CMSOi transduction is defined analogously.
\end{defi}

Note, once again, that the result of a transduction $\tau$ of one graph is
generally a set of graphs, due to the involved expansion map.
For a graph class $\cf H$, the {\em transduction $\tau$ of
the class $\cf H$} is the union of the particular transduction results,
precisely, $\tau(\cf H):=\bigcup_{G\in\cf H}\tau(G)$.

\section{Capturing Height of Dense Graphs}
\label{sec:shrub}

The concept of tree-depth is commonly used to capture the ``height'' 
of other graphs than just trees.
Actually, tree-depth can be seen as a bounded-height analogue of tree-width.
However, as discussed already in the introduction,
the main drawback of tree-depth (as well as of tree-width) is its incapability 
to handle dense graphs and some simple graph operations like the complement.
Since, on the other hand, clique-width handles dense ``uniform'' graphs and
the complement operation smoothly, it makes good sense to try to modify
its definition towards capturing ``height'' in addition to ``width''.

Unfortunately, such a direct modification of clique-width seems not possible,%
\footnote{For example, simply trying to restrict the underlying expressing
 tree in Definition~\ref{def:cw} brings the necessity of disjoint unions of
 an arbitrary arity which, in turn, ``weakens'' the definition too much.
 This is precisely the point at which the NLC approach
 (Subsection~\ref{sub:widths}) with explicitly adding edges only between the
 graphs participating in a disjoint union operation turns out better.}
and one has to look at other related width measures, namely to the so called 
neighbourhood diversity and the aforementioned NLC-width for an inspiration.

Before we continue, notice that the requirement to smoothly handle dense
graphs and the graph complement operation, naturally means that a
new measure cannot be stable under taking non-induced subgraphs.

\subsection{Shrub-depth}

To motivate the coming definition of shrub-depth, we recall the {neighbourhood
diversity} parameter introduced by Lampis~\cite{lam12} in an algorithmic context:
Two vertices $u,v$ are {\em twins} in a graph $G$ if $N_G(u)\sem\{v\}=N_G(v)\sem\{u\}$.
The {\em neighbourhood diversity} of $G$ is the smallest $m$ such
that $V(G)$ can be partitioned into $m$ sets such that in each part the
vertices are pairwise twins.
This basically means that $V(G)$ can be coloured by $m$
exclusive labels such that the existence of an edge $uv$ depends 
solely on the colours of $u$ and~$v$.

To stress that the considered labels are exclusive, we shall instead 
call them {\em colours}.
Inspired by attempts to generalize neighbourhood diversity, e.g, 
in \cite{Ganian15,gaj12}, we come to the idea of enriching the
diversity colouring with a bounded number of ``layers''.
This results in the following formalization:

\begin{defi}[Tree-model]
\label{def:tree-model} 
Let $m$ and $d$\, be non-negative integers. 
A {\em tree-model of $m$ colours and depth $d$} for a graph $G$ is a pair
$(T,S)$ of a rooted tree $T$ (of~height~$d$) and a set
$S\subseteq\{1,2,\ldots,m\}^2\times \{1,2,\ldots,d\}$ (called a {\em
signature} of the tree-model) such that 
\begin{enumerate}
\item the length of each root-to-leaf path in $T$ is exactly $d$,
\item the set of leaves of $T$ is exactly the set $V(G)$ of vertices of $G$,
\item each leaf of $T$ is assigned one of the colours $\{1,2\dots,m\}$, 
and
\item\label{it:tree-model-edge} 
for any $i,j,\ell$ it holds $(i,j,\ell)\in S$ iff $(j,i,\ell)\in S$
(symmetry in the colours),
and for any two vertices $u,v\in V(G)$ such that $u$ is coloured $i$ and $v$
is coloured $j$ and the distance between $u,v$ in $T$ is $2\ell$,
the edge $uv$ exists in $G$ if and only if $(i,j,\ell)\in S$.
\end{enumerate}
Note that point \eqref{it:tree-model-edge} effectively says that the existence of
a $G$-edge between $u,v\in G$ depends solely on the colours of $u,v$ and
the depth of the least common ancestor $u\wedge v$ in $T$. 
We hence, for convenience, call $T$ itself a {\em tree-model of~$G$}, 
assuming that the signature set $S$ is implicitly associated with~$T$.

The class of all graphs having such a tree-model of $m$ colours and depth $d$ 
is denoted by $\TM dm$.
\end{defi}

\begin{figure}[tb]
\begin{center}
\scalebox{1}[-1]{\includegraphics[width=.6\hsize]{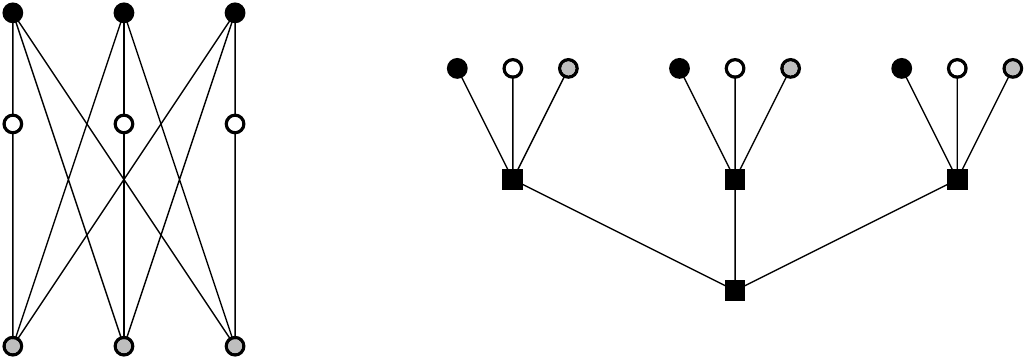}}
\end{center}
\caption{The graph obtained from $K_{3,3}$ by subdividing a matching belongs to
$\TM 23$.}
\label{fig:TMillustration}
\end{figure}

For instance, $K_n\in\TM11$ and $K_{n,n}\in\TM12$.
More generally, $\TM1m$ is exactly the class of graphs of neighbourhood
diversity at most~$m$.
For a more involved example, imagine an arbitrarily large collection of 
graphs $G_1,G_2,\dots,G_s\in\TM1m$, such that $\{G_1,\dots,G_s\}$ is
partitioned into $m'$ groups.
Let $H$ be a graph obtained from a disjoint union $G_1\cup\dots\cup G_s$
by adding, say, all edges between distinct graphs from the groups 1 and 3,
all edges from graphs in the group 2 to graphs in the groups 5 and 7, etc.
Then $H\in\TM2{m\cdot m'}$.
This ``hierarchical'' example can be easily generalized to higher values of~$d$.
Yet another illustrations can be found in Figures~\ref{fig:TMillustration}
and~\ref{fig:P8tm}.

It is easy to see that each class $\TM dm$ is closed under complements 
and induced subgraphs (which is our desire), 
but neither under disjoint unions, nor under subgraphs.
If $G$ has a tree-model $T$ and $H$ is any induced subgraph of~$G$,
then the corresponding induced subtree of $T$ immediately gives a tree-model for~$H$.
Note also that one coloured tree $T$ can be a tree-model of
several graphs (on the same vertex set), depending on the associated signatures.

Another interesting observation is the relation of a tree-model to a certain
generalization of the NLC classes from Subsection~\ref{sub:widths}:
imagine that the definition of NLC$_m$ is allowed to make disjoint union of
an arbitrary number of graphs (but still with a uniform rule for adding
edges between them), and the depth of the construction tree is bounded by~$\leq d$.
If we, furthermore, forbid the relabelling operation, then the result coincides
with the class $\TM dm$.
Even if relabellings are allowed in NLC$_m$, 
we can encode all label changes in the leaf colours
thanks to the bounded depth of the construction (at the price of increasing~$m$).

The depth of a tree-model generalizes tree-depth of a graph as
follows (while the other direction is obviously unbounded, e.g., for cliques):

\begin{prop}
\label{pro:td-to-tree-model}
If $G$ is of tree-depth $d$, then $G\in\TM{d}{2^{d}}$.
If, moreover, $G$ is connected, then also $G\in\TM{d-1}{2^{d}}$.
\end{prop}

\begin{proof}
Let $U$ be an inclusion-minimal rooted forest of height $d-1$ such that 
$G\subseteq\Clos(U)$, and let $T$ be a rooted tree obtained by 
adding a new root $r$ connected to the former roots of $U$, and $d'=d$.
If $G$ is connected, then $U$ already is a tree, and then we set $T=U$ and
$d'=d-1$.

For $u\in V(T)$ we set a colour $c(u)=(j,I)$ such that $dist_T(r,u)=d'-j$
and $I=\{i: \{u,anc_i(u)\}\in E(G)\}$, where $anc_i(u)$ denotes the ancestor
of $u$ in $T$ at distance $i$ from~$u$.
Notice that $I\subseteq\{1,\dots,d-1-j\}$ (because of the height of $U$),
and so the total number of distinct
$c(u)$ over all $u\in V(U)$ is $2^{d-1}+2^{d-2}+\dots+1<2^{d}$.
Let $T'$ be obtained from $T$ as follows:
For every node $u\in V(U)$ such that $dist_T(r,u)<d'$, 
we add to $u$ a new path with the other end denoted by $u'$ 
such that $dist_{T'}(r,u')=d'$, and set $c(u')=c(u)$.

We claim that this $T'$ with the colours $c(v)$ in the leaves of $T'$ is the
desired tree-model of $G$.
Let $G'$ be the graph defined on the leaves of $T'$ as follows;
$\{u,v\}\subseteq V(G')$ is an edge of $G'$ iff, for $c(u)=(j_1,I_1)$,
$c(v)=(j_2,I_2)$ and $j_1<j_2$, it holds $dist_{T'}(u,v)=2j_2$
and $j_2-j_1\in I_1$.
Then clearly $G'\simeq G$.
\end{proof}

When dealing with tree-models of graph classes 
(e.g., in model checking or in transductions), 
the depth parameter $d$ is asymptotically much more important than the number of colours $m$.
With this in mind, it is useful to work with a more streamlined notion which
only requires a single parameter $d$, and to this end, we introduce the following:

\begin{defi}[Shrub-depth]
\label{def:shrub-depth}
A class $\cf G$ of graphs has {\em shrub-depth} $d$
if there exists $m$ such that $\cf G\subseteq\TM dm$,
while for all natural $m'$ it is the case that $\cf G\not\subseteq\TM{d-1}{m'}$.
In a wider sense, $\cf G$ is of {\em bounded shrub-depth}
if there exist integers $d,m$ such that $\cf G\subseteq\TM dm$.
\end{defi}
Note that Definition~\ref{def:shrub-depth} is asymptotic as it makes sense 
only for infinite graph classes; the shrub-depth of a single finite graph is 
always at most one ($0$ for empty or one-vertex graphs).
Furthermore, it makes no sense to say ``the class of all graphs of shrub-depth~$d$''.

For instance, the class of all cliques has shrub-depth $1$.
On the other hand, it will follow from Theorem~\ref{thm:Path-TM} 
that the class of all paths has unbounded (infinite) shrub-depth.
Now we argue that this new notion is indeed ``intermediate'' between
tree-depth and clique-width (and even linear clique-width).

\noindent\begin{minipage}{\textwidth}
\begin{prop}
\label{pro:td-shrub-cw}
Let $\cf G$ be a graph class and $d$ an integer. Then:
\begin{enumerate}[label={\it\alph*)}]
\item
If $\cf G$ is of tree-depth $\leq d$, then $\cf G$ is of shrub-depth~$\leq d$.
\item
If $\cf G$ is of bounded shrub-depth, then $\cf G$ is of bounded
linear clique-width.
\end{enumerate}
The converse statements are not true in general.
\end{prop}
\end{minipage}

\begin{proof}
a) This follows from Proposition~\ref{pro:td-to-tree-model},
and the converse cannot be true in general because of, e.g., the class of
all cliques.

b) We remark that it is trivial to see that $\cf G$ is of bounded clique-width.
Here we even show how to straightforwardly translate a tree-model
with $m$ colours and depth $d$ into a linear (caterpillar-shaped) $m(d+1)$-expression:
Let $v_1,\dots,v_n$ be any (usual) left-to-right ordering of the leaves of a
tree-model $T$ of some $G\in {\cf G}$.
The expression is constructed inductively for $i=1,\dots,n$ as follows:
\begin{itemize}
\item a vertex $v_i$ is created and added with a (currently unique) colour $(c,0)$
where $c=c(u_i)$ is its colour in $T$,
\item whenever colour $c$ is to be adjacent to colour $c'$ at distance $2d$ in
the model $T$, the expression adds all edges between the colours $(c,0)$ and
$(c',d)$, and
\item for $2d'$ being the distance from $v_i$ to $v_{i+1}$ in $T$,
the expression changes all colours $(c,d)$ with $d<d'$ to $(c,d')$.
\end{itemize}
A counterexample to the converse claim is, e.g., the class of all paths
by Theorem~\ref{thm:Path-TM}.
\end{proof}

The relation between classes of bounded shrub-depth and of bounded
tree-depth is even deeper than shown above.
The operation of a {\em local complementation} in a graph takes any vertex
$v$ and replaces the subgraph induced on the neighbours of $v$ with its edge-complement.
A graph $G$ is a {\em vertex-minor} of a graph $H$ if $G$ is an induced
subgraph of a graph $H'$ such that $H'$ is obtained from $H$ by a sequence
of {local complementations}.
As shown in \cite{hkoo16}, the class of vertex-minors of all graphs of
tree-depth at most $d$ has shrub-depth at most~$d$,
and every class of shrub-depth $d$ can be constructed as vertex-minors of
graphs of tree-depth $d'$ where $d'$ depends (only) on~$d$.

\subsection{SC-depth}

A significant drawback of the notion of shrub-depth is the aforementioned
fact that it does not make sense to ask about the shrub-depth of 
a single finite graph.
Here we propose a remedy for this problem in the form of another, very simple
and single-parameter based, definition of a depth-like parameter 
which turns out to be asymptotically equivalent to shrub-depth.
(Although, several years of research experience since \cite{ghnoor12} have
also shown many clear advantages of the shrub-depth notion.)

Let $G$ be a graph and let $X\subseteq V(G)$.
We denote by $\overline{G}^X$ the graph $G'$ with vertex set $V(G)$ where
$x\neq y$ are adjacent in $G'$ if 
(i) either $\{x,y\}\in E(G)$ and $\{x,y\}\not\subseteq X$, or
(ii) $\{x,y\}\not\in E(G)$ and $\{x,y\}\subseteq X$.
In other words, $\overline{G}^X$ is the graph obtained from $G$ 
by complementing the edges on~$X$.

\begin{defi}[SC-depth\footnote{As the ``Subset-Complementation'' depth.}]
\label{def:SC-depth}
We define inductively the class $\SC n$ as follows:
\begin{itemize}
  \item We let $\SC0=\{K_1\}$;
  \item if $G_1,\dots,G_p\in\SC n$ and $H= G_1\dot\cup\dots\dot\cup G_p$
  denotes the disjoint union of the $G_i$'s,
  then for every subset $X$ of vertices of $H$, we
  have $\overline{H\,}^X\in\SC{n+1}$.
\end{itemize}
The {\em SC-depth} of $G$ is the minimum integer $n$ such that $G\in\SC n$.
\end{defi}

\begin{figure}[tb]
\begin{center}
\includegraphics[width=0.98\hsize]{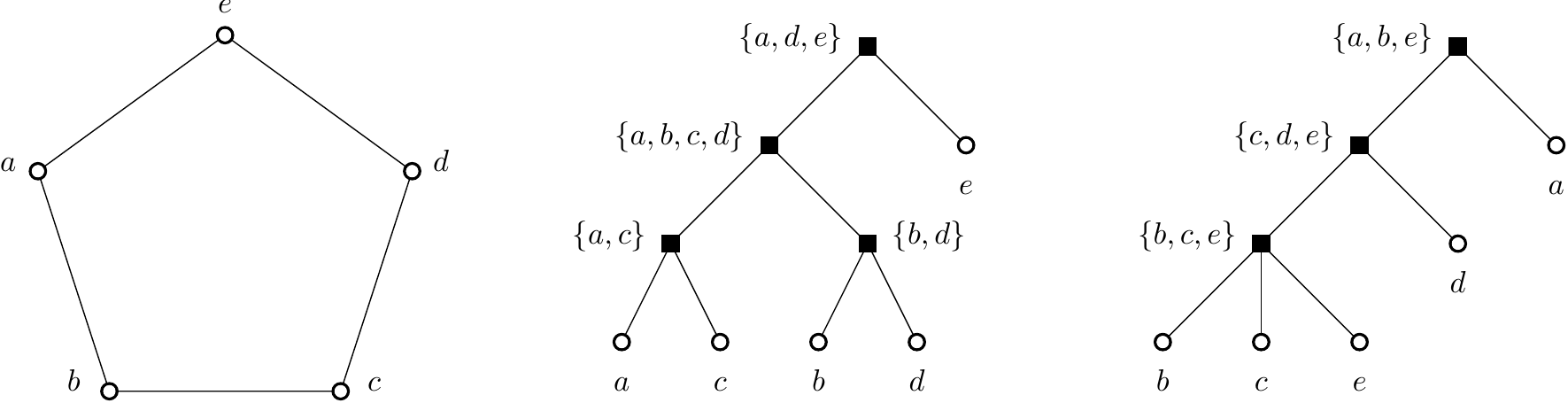}
\caption{A graph $G$ and two possible SC-depth representations by
depicted trees.}
\label{fig:kdtree}
\end{center}
\end{figure}

The SC-depth of a graph $G$ is thus the minimum height of a rooted tree~$Y$,
such that the leaves of $Y$ form the vertex set of $G$, and each internal
node $v$ is assigned a subset $X$ of the descendant leaves of $v$. 
Then the graph corresponding to $v$ in $Y$ is the complement on $X$, of the disjoint
union of the graphs corresponding to the children of $v$ (see Figure~\ref{fig:kdtree}).

The reason we introduce both the asymptotically equivalent SC-depth and shrub-depth 
measures here is that each brings a unique perspective on the classes of graphs 
we are interested in (see e.g.~\cite{hkoo16}).

\begin{thm}
\label{thm:shrubsc}
Let $\cf G$ be a class of graphs. Then the following are equivalent:
\begin{itemize}
  \item \mbox{There exist integers $d$,$m$ such that $\cf G\subseteq \TM dm$ 
	($\cf G$ has bounded shrub-depth).}
  \item There exists an integer $k$ such that $\cf G\subseteq \SC k$ 
	($\cf G$ has bounded SC-depth).
\end{itemize}
More precisely, $\TM dm\subseteq \SC{dm(m+1)}$
and $\SC k\subseteq \TM k{2^k}$.
\end{thm}

\begin{proof}
We prove the forward implication by induction on $d$.
In the degenerate base case $d=0$, it is trivially the case that $\TM 0m=\{K_1\}=\SC 0$.
Assume now $G\in \TM{d+1}m$ for some $d\geq0$.
By Definition~\ref{def:tree-model}, there exist
an integer $c\geq1$ and graphs $G_1,\dots,G_c\in\TM dm$
(actually subgraphs of $G$ induced by the leaf sets of the root-subtrees in the
respective tree-model of $G$) such that the following holds:
$G$ results from the disjoint union $G_1\cup\dots\cup G_c$
by adding those edges $uv$ for which $u$ and $v$ belong to distinct graphs
among $G_1,\dots,G_c$, and for the pair $i,j$ of colours of $u,v$,\,
$(i,j,d+1)$ belongs to the signature~$S$.

By the induction assumption, we have got $G_1,\dots,G_c\in\SC{k_0}$ 
for some integer~$k_0$.
For each of these graphs $G_\ell$, $\ell\in\{1,\dots,c\}$,
we successively complement the edges on the following subsets of vertices:
\begin{itemize}
\item for each $i$ such that $(i,i,d+1)\in S$,
on the set $X_\ell^i\subseteq V(G_\ell)$ of the vertices of $G_\ell$ of colour~$i$,
\item for each $i<j$ such that $(i,j,d+1)\in S$, on the set $X_\ell^i\cup X_\ell^j$
(defined as above), then on the set $X_\ell^i$ itself and then on $X_\ell^j$ itself.
\end{itemize}
Observe that at most $m+3{m\choose2}$ complement operations are applied to
each $G_\ell$, and this number can be reduced down to
$m+{m\choose2}={m+1\choose2}$ by skipping possible repeated complements.
Denoting by $G'_\ell$ the graph obtained in this way from $G_\ell$
we get, by Definition~\ref{def:SC-depth}, that
$G'_1,\dots,G'_c\in\SC{k_1}$ where $k_1=k_0+{m+1\choose2}$.

Effectively, in each $G_\ell$ we have complemented the edges whose colour
pairs (together with third $d+1$) belong to~$S$.
In the next step we make the disjoint union $G':=G'_1\cup\dots\cup G'_c$
and repeat the same complementation procedure on this global level.
Namely:
\begin{itemize}
\item for each $i$ such that $(i,i,d+1)\in S$, on the set $X^i\subseteq V(G')$ 
of the vertices of $G'$ of colour~$i$,
\item for each $i<j$ such that $(i,j,d+1)\in S$, on the set $X^i\cup X^j$,
then on $X^i$ and then on~$X^j$.
\end{itemize}
Denoting the resulting graph by $G''$, we similarly get
$G''\in\SC{k_2}$ where $k_2=k_1+{m+1\choose2}=k_0+m(m+1)$.
It remains to routinely verify that $G''\simeq G$.

\smallskip
As for the backward implication, we directly construct a tree-model for each
graph $G\in \SC k$.
By Definition~\ref{def:SC-depth}, $G\in \SC k$ can be constructed along
a rooted tree $T$ such that the leaf set of $T$ is $V(G)$
and each internal node $t$ of $T$ is associated with a complement set
$X_t$ (which is a subset of the descendant leaves).
We assign the leaf colours as follows.
Let $v\in V(G)$ be a leaf of $T$, and $t_0=v,t_1,\dots,t_k=r$ be the path from 
$v$ to the root $r$ of $T$.
We colour $v$ with the binary vector $(a_i)_{i=1}^k$ such that
$a_i=1$ iff $v\in X_{t_i}$.

By Definition~\ref{def:SC-depth}, $uv$ forms an edge of $G$, if and only if
the pair $\{u,v\}$ belongs to an odd number of the complement sets $X_t$
over the whole~$T$.
This can easily be determined from the colours of $u$ and $v$, and from the
depth of their least common ancestor in~$T$.
Consequently, $G\in \TM k{2^k}$.
\end{proof}

\subsection{Long paths}

For graphs of small tree-depth a characteristic property is the absence of
long paths as subgraphs, cf.~Proposition~\ref{prop:td-basics}\,b).
This is obviously false for classes of small shrub-depth since those, in
particular, include all cliques and bicliques.
However, one can restrict induced paths in every class $\TM dm$,
as follows.

\begin{thm}\label{thm:Path-TM}
Let $\ell=3\cdot 2^m-4$ and $P_\ell$ denote the path of length $\ell$,
i.e., on $\ell+1$ vertices.
Then $P_\ell\in \TM {2m+1}m$, but for any $d$ we have
$P_{\ell+1}\not\in \TM dm$.
\end{thm}
In particular, there exist no $d,m$ such that $\TM dm$ would contain all
paths.

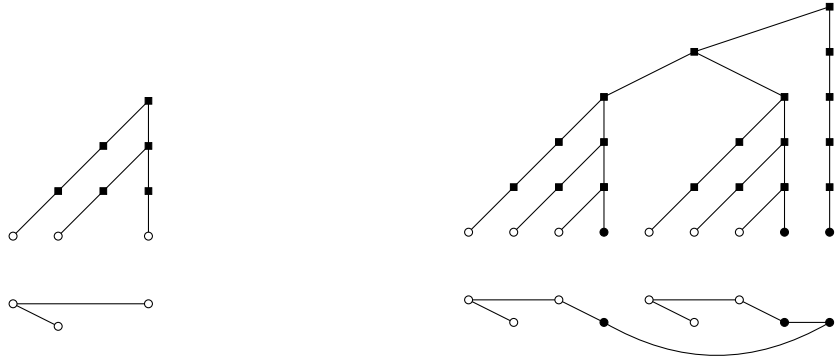
\begin{figure}[ht]
{\hfill
\begin{tikzpicture}[scale=0.6]
\tikzstyle{every node}=[draw, shape=rectangle, minimum size=2.5pt,inner sep=1pt, fill=black]
\tikzstyle{cbl}=[shape=circle, fill=black, minimum size=3pt]
\tikzstyle{cwh}=[shape=circle, fill=white, minimum size=3pt]
\draw (2,4) node {} -- (2,3) node {} -- (2,2) node {} -- (2,1) node[style=cwh] {};
\draw (2,4) -- (1,3) node {} -- (0,2) node {} -- (-1,1) node[style=cwh] {};
\draw (2,3) -- (1,2) node {} -- (0,1) node[style=cwh] {};
\tikzstyle{every node}=[draw, style=cwh, inner sep=1pt]
\draw  (0,-1) node {} -- (-1,-0.5) node {} -- (2,-0.5) node {} ;
\draw  (0,-1.8) node[color=white] {};
\end{tikzpicture}
\hfill\hfill
\begin{tikzpicture}[scale=0.6]
\tikzstyle{every node}=[draw, shape=rectangle, minimum size=2.5pt,inner sep=1pt, fill=black]
\tikzstyle{cbl}=[shape=circle, fill=black, minimum size=3pt]
\tikzstyle{cwh}=[shape=circle, fill=white, minimum size=3pt]
\draw (8,6) node {} -- (8,5) node {} -- (8,4) node {} -- (8,3) node {}
	 -- (8,2) node {} -- (8,1) node[style=cbl] {};
\draw (8,6) -- (5,5) node {} -- (3,4) node {} -- (2,3) node {} -- (1,2) node {}
	 -- (0,1) node[style=cwh] {};
\draw (3,4) -- (3,3) node {} -- (2,2) node {} -- (1,1) node[style=cwh] {};
\draw (3,3) -- (3,2) node {} -- (2,1) node[style=cwh] {};
\draw (3,2) node {} -- (3,1) node[style=cbl] {};
\draw (5,5) -- (7,4) node {} -- (6,3) node {} -- (5,2) node {}
	 -- (4,1) node[style=cwh] {};
\draw (7,4) -- (7,3) node {} -- (7,2) node {} -- (6,1) node[style=cwh] {};
\draw (7,3) -- (6,2) node {} -- (5,1) node[style=cwh] {};
\draw (7,2) node {} -- (7,1) node[style=cbl] {};
\tikzstyle{every node}=[draw, style=cwh, inner sep=1pt]
\draw  (1,-1) node {} -- (0,-0.5) node {} -- (2,-0.5) node {}
	-- (3,-1) node[style=cbl] {} to [out=-30,in=-150] (8,-1) ;
\draw  (5,-1) node {} -- (4,-0.5) node {} -- (6,-0.5) node {}
	-- (7,-1) node[style=cbl] {} -- (8,-1) node[style=cbl] {} ;
\end{tikzpicture}
\hfill}
\medskip
\caption{Tree-models (top) for small paths (below), 
	as used in the proof of	Proposition~\ref{thm:Path-TM}.
	Left: with $1$ colour and depth $3$ for $P_2$.
	Right: $2$ colours and depth $5$ for~$P_8$.}
\label{fig:P8tm}
\end{figure}

\begin{proof}
We start with the construction of $P_\ell\in \TM {2m+1}m$
that is, of an appropriate tree-model $T_m$ of~$P_\ell$,
by induction on~$m$.
We shall maintain a special property that each end of $P_\ell$
is represented in $T_m$ by a leaf which has no siblings, 
i.e., its parent is of degree~$2$.
As the base case, we use the tree-model $T_1$ of $m=1$ colours and depth $2m+1=3$
from the left-hand side of Figure~\ref{fig:P8tm}.
(Note that although $P_2\in\TM21$, we use an extra level in $T_1$ 
to achieve our property.)

We now construct $T_{m+1}$ for $m\geq1$.
Let $u$ and $v$ be the ends of $P_\ell$, and recall that
each of $u,v$ has no siblings in $T_m$.
We create a sibling $u_1$ of $u$ in $T_m$ and assign $u_1$ a new colour $m+1$. 
This intermediate tree-model $U_m$ can represent $P_{\ell+1}$
with the ends $u_1,v$ and, see Figure~\ref{fig:P8tm} right,
the desired model $T_{m+1}$ follows:
\begin{itemize}
\item for $U_m$ and its disjoint copy $U_m'$, 
add a common ancestor $q$ of their roots,
\item create a rooted path of length $2m+3$, with the root $r$ and the only
leaf $w$ of colour $m+1$, and make $q$ another son of~$r$.
\end{itemize}
Clearly, $T_{m+1}$ is a tree-model of $m+1$ colours and depth
$2m+3=2(m+1)+1$, and it can represent the edges $u_1w$ and $u_1'w$ but not
$u_1u_1'$.
Thus $T_{m+1}$ makes a tree-model of $P_{\ell'}$ for 
$\ell'=2(\ell+1)+2=3\cdot 2^{m+1}-4$.

\smallskip

In the converse direction we start with an easy observation for
$m=1$; $P_3\not\in\TM d 1$ for any~$d$ (this follows from the folklore
fact that the path on $4$ vertices is not a cograph, too).
The proof can then be finished by induction over $m\geq1$,
provided that we establish the following contrapositive claim:
if $P_{2\ell+5}\in\TM d{m+1}$ for any $\ell,m,d\geq1$, then $P_{\ell+1}\in\TM dm$.

So fix $\ell$ and $m$, 
and assume $G:=P_{2\ell+5}\in\TM d{m+1}$ and $T$ is a corresponding tree-model
of $m+1$ colours and minimum possible height~$d$.
In this proof we denote by $T_x$ the subtree of $T$ rooted at a node~$x$.
As $d$ is minimum and $P_{2\ell+5}$ is connected, 
there exist distinct sons $u,v$ of the root of $T$
and colours $i,j$ (possibly equal),
such that $T_u$ includes at least one leaf with colour $i$
and $T_v$ at least one leaf with colour $j$,
and the colour pair $(i,j)$ at distance $2d$ determines an edge.

We let $J\subseteq G$ denote the subgraph formed only by those edges which
are determined by the colour pair $(i,j)$ at distance $2d$ in $T$,
i.e., $xy\in E(J)$ iff the colours of $x,y$ are $i,j$ in $T$ and the only
common ancestor of $x,y$ is the root of~$T$.

If $i=j$, then we claim that there cannot be two non-incident edges in~$J$.
Indeed, this would necessarily mean that $J$ contains $K_{2,2}$, but
$K_{2,2}\not\subseteq G$.
Hence $J$ is $K_2$ or $K_{1,2}$ and there exist at most three vertices of
colour $i$ altogether, and in either case one subpath in $G-V(J)$
is of length at least $\lceil(2\ell+5-4)/2\rceil=\ell+1$.
Hence $T-V(J)$ gives a tree-model of $P_{\ell+1}$ of $m$ labels.

We now examine the other possibility $i\not=j$.
First, we observe that if $x_1y_1,x_2y_2$ are non-incident edges of $J$ 
such that $x_1,x_2$ are of the same colour, then the only   
common ancestor of $x_1,x_2$ is the root of~$T$.
Otherwise, we would get a contradiction that $K_{2,2}\subseteq J$.
Second, we argue that there cannot be three pairwise non-incident edges
$x_1y_1,x_2y_2,x_3y_3$ in $J$ (where $x_1,x_2,x_3$ are of the same colour).
If this happened, then (say) for the vertex $y_1$ at least two
of the vertices~$x_1,x_2,x_3$ would have only one common ancestor with $y_1$, 
the root of~$T$.
Consequently, $y_1$ would have at least two neighbours in the set
$\{x_1,x_2,x_3,y_1,y_2,y_3\}$, and the same would symmetrically hold for
all the members of this set, contradicting the fact that $J$ is acyclic.

Therefore, $J$ is a path of length at most $4$,
or $J$ consists of two components isomorphic to $K_2$ or $K_{1,2}$.
Moreover, if there exists a leaf $z$ of colour $i$ or $j$ in $T$ which is not
incident to an edge of $J$, then $J$ has no two non-incident edges
and all such leaves (of colour $i$ or $j$) not incident to $E(J)$ 
are of the same colour, as can be easily checked.

We first consider the case that $J$ has one component.
If it is $K_2$ or $K_{1,2}$ then,
by the previous arguments, all the leaves of $T$ coloured $i$ (say) are incident to the
one or two edges of~$J$.
As above (in the case of $i=j$) we can now argue that $T-V(J)$ gives a
tree-model of $P_{\ell+1}$ of $m$ labels.
If, on the other hand, $J$ is $P_3$ or $P_4$,
then all the leaves of $T$ coloured $i$ or $j$ are incident to the edges of~$J$.
We form a new tree-model $T'$ by removing from $T$ all the leaves
of colours $i,j$ (i.e., incident to the edges of $J$) 
and adding arbitrarily one new leaf of colour $i$.
Then $T'$ of $m$ labels models a path $P_{2\ell+1}$ (or $P_{2\ell+2}$).

We are left with the case of $J$ consisting of two components, such that all
the leaves of $T$ coloured $i$ or $j$ are incident to the edges of~$J$.
If any of the subpaths of $G-V(J)$ is of length at least $\ell+1$, then we
are again done.
Otherwise, we can choose one component $J_1$ of $J$ such that
$G-V(J_1)$ contains a subpath $G'$ of length at least $\ell+3$.
We denote by $J_2$ the other component of $J$ (presumably $J_2\subseteq G'$),
and form a new tree-model $T'$ by restricting $T$ to the leaves from $G'$,
removing the leaves of $J_2$ and adding arbitrarily one leaf of colour $i$
(recall that no vertex of $G'-V(J_2)$ has colour $i$ or $j$).
Hence $T'$ of $m$ labels models a path $P_{\ell+1}$ (or $P_{\ell+2}$).
\end{proof}

The combinatorial result in Theorem~\ref{thm:Path-TM} has interesting
relations also to logical questions (see Section~\ref{sec:transdhier}).
For instance, in respect of the research of \MSO-orderable graphs
by Blumensath and Courcelle~\cite{bc14},
note that in the class of all finite paths one can easily define
a linear ordering by an \MSOi formula.
Hence it immediately follows from a characterization given in 
\cite[Theorem~5.31]{bc14} 
that the class of all finite paths cannot have bounded shrub-depth.
The advantage of our Theorem~\ref{thm:Path-TM} (occurirng already
in~\cite{ghnoor12}) is that it gives exact combinatorial bounds.
Furthermore, Theorem~\ref{thm:Path-TM} together with
Theorem~\ref{thm:shrub-depth-interpretability} implies
the result~\cite[Theorem~5.31]{bc14} that infinite graph classes 
of bounded shrub-depth are not \MSOi-orderable.

\begin{figure}[tb]
\begin{center}
\includegraphics[width=.6\textwidth]{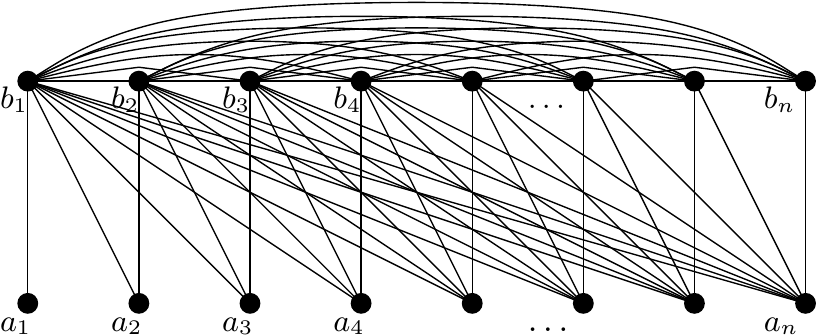}
\caption{An example of a graph class not containing any induced subpaths 
of length~$3$, which has unbounded shrub-depth.
In fact, these graphs are even the so called threshold graphs (a~special case of small linear
clique-width) -- view the vertices in the backward order
$a_n,b_n,a_{n-1},b_{n-1},\dots,a_1,b_1$.}
\label{fig:halft}
\end{center}
\end{figure}

Note, however, that graph classes of bounded
shrub-depth are not asymptotically related to those excluding long induced subpaths;
in the opposite direction the situation here is very different 
than in Proposition~\ref{prop:td-basics}\,b).
As an example, we mention the graph class from Figure~\ref{fig:halft}
which contains no induced subpaths of length~$3$.
One can give a direct combinatorial proof that this class is of unbounded
shrub-depth (similarly as for Theorem~\ref{thm:Path-TM}), 
but we skip it here since this fact follows from the aforementioned result of~\cite{bc14}
(the graph of Figure~\ref{fig:halft} is FO-orderable)
or, alternatively, from a combination of results of~\cite{hkoo16}.

\subsection{Induced subgraphs characterization}

Lastly in this section,
we provide yet another characterization of the classes defined previously.
In a nutshell, we are going to show that each of these classes can be
characterized by a finite list of {\em forbidden induced subgraphs}.
A nice consequence of this finding is that membership in each of the classes
can be tested in polynomial time.
The tool we use here is well-quasi-ordering.

A class or property is said to be {\em hereditary}
if it is closed under taking induced subgraphs.
A {\em well-quasi-ordering} (or {\em wqo}) of a set $X$
is a quasi-ordering on $X$ such that for any infinite sequence of elements 
$x_1, x_2,\dots$ of $X$ there exist $i<j$ with  $x_i\leq x_j$. In other words,
a wqo is a quasi-ordering that does not contain an infinite strictly decreasing
sequence or an infinite set of incomparable elements.
We are going to use the following well-known result:

\begin{thm}[Ding~\cite{din92}]
\label{thm:Ding}
Let $m\in\mathbb N$ be an integer and $C$ be a finite set of colours.
The class of the graphs not containing a path on $m$ vertices as a
subgraph and with vertices coloured by $C$ is well-quasi-ordered
under the colour-preserving induced subgraph order.
\end{thm}

\begin{cor}
\label{cor:treemodelWQO}
Let $\cf S$ be a graph class of bounded shrub-depth, such that the vertices of the graphs in $\cf S$
are coloured from a finite set $C$ of colours.
Then $\cf S$ is well-quasi-ordered under the colour-preserving induced subgraph order.
\end{cor}
\begin{proof}
Consider an infinite sequence $(G_1,G_2,\dots)\subseteq\cf S$, and the
corresponding tree-models $(T_1,T_2,\dots)$.
Let $T_i^+$, $i=1,2,\dots$, denote the rooted tree with leaf labels composed
of the colours of $T_i$ and the colours of $G_i$.
By Theorem~\ref{thm:Ding},
$T_1^+,T_2^+,\dots$ of bounded diameter is wqo under the 
rooted coloured subtree relation, and, consequently, so are
the coloured graphs $G_1,G_2,\dots$, as desired.
\end{proof}

The advertised result now follows by a simple twist as follows.

\begin{thm}
\label{thm:clforbidden}
For every integers $d,m$, there exists a finite set of graphs $\cf F_{d,m}$
(the forbidden subgraphs)
such that a graph $G$ belongs to $\TM dm$ if and only if $G$ has no induced
subgraph isomorphic to a member of $\cf F_{d,m}$.

Similarly, for every $n$ there exists a finite set of graphs $\cf F'_{n}$
such that $G\in\SC n$ if and only if $G$ has no induced
subgraph isomorphic to one of $\cf F'_{n}$.
\end{thm}

\begin{proof}
We let $\cf F_{d,m}$ be the (isomorphism-free) set of graphs $H$ such that
$H\not\in\TM dm$ but $H-v\in\TM dm$ for every $v\in V(H)$.
By this definition, no member of $\cf F_{d,m}$ is a proper induced subgraph
of another member.
Hence it is enough to argue that $\cf F_{d,m}$ is wqo to conclude that 
$\cf F_{d,m}$ is finite.

The latter follows from an easy observation: if $H-v\in\TM dm$ for some
$v\in V(H)$, then $H\in\TM d{2m+1}$.
Indeed, we take a tree-model of $H-v$, add arbitrarily a new leaf
of a unique new colour for~$v$
and annotate with an extra bit the colours of all leaves which are
neighbours of $v$ in~$H$.
The result is a tree-model for $H$ with $2m+1$ colours.
Consequently, $\cf F_{d,m}\subseteq \TM d{2m+1}$ and the wqo property
follows from Corollary~\ref{cor:treemodelWQO}.

The second claim is proved analogously.
We let $\cf F'_{n}$ be the (isomorphism-free) set of graphs $H$ such that
$H\not\in\SC n$ but $H-v\in\SC n$ for every $v\in V(H)$.
By Theorem~\ref{thm:shrubsc},
$\{H-v:H\in\cf F'_{n}\}\subseteq\TM n{2^n}$, and so
$\cf F'_{n}\subseteq\TM n{2^{n+1}+1}$ by the previous paragraph.
The wqo property again follows from Corollary~\ref{cor:treemodelWQO}.
\end{proof}

The ``obstacle'' sets $\cf F_{d,m}$ and $\cf F'_{n}$ of
Theorem~\ref{thm:clforbidden} are not only of mathematical interest, but
also have algorithmic consequences.
Namely, in connection with established algorithms they allow for efficient
membership testing of these classes.
Note, however, that we do not provide an algorithmic construction of the sets
$\cf F_{d,m}$ and $\cf F'_{n}$, and so we only prove an existence of the
respective algorithms for each specific values of $d,m$ and $n$ 
(in parameterized complexity theory this is formally called {\em nonuniform FPT}).

\begin{cor}
The problems to decide, for a given graph $G$, whether $G\in\TM dm$ and
whether $G\in\SC n$, are fixed-parameter tractable with respect to the
parameters $d,m$ and $n$, respectively.
\end{cor}
\begin{proof}
We provide a proof for the problem of $G\in\TM dm$, while that of 
$G\in\SC n$ is very similar.
As mentioned before, the class $\TM dm$ is of bounded clique-width
(namely, $2m$ is a trivial upper bound).
Therefore, one can use \cite{ho08} to compute in FPT an approximate
expression of $G$ of clique-width depending only on $m$ or to
correctly conclude that $G\not\in\TM dm$.
In the former case, one can then call the algorithm of \cite{cmr00}
to test whether any member of $\cf F_{d,m}$ is an induced subgraph of~$G$.
Based on the outcome, the correct decision
about $G\in\TM dm$ is easily made.
\end{proof}

\section{Shrub-depth and \MSO Transductions}
\label{sec:transdhier}

While in the previous section we have focused on establishing basic
combinatorial properties of shrub-depth and SC-depth,
now we shift our attention towards their logical aspects.
The final outcome will be the finding that (a slight technical adjustment of)
tree-models of depth $d$ precisely capture the $d$-th finite level of the \MSOi
transduction hierarchy of simple undirected graphs,
for all $d\in\mathbb N$.
For that, we start by showing that shrub-depth indeed goes well with simple
\MSOi interpretations.

\subsection{Stability under interpretations}
\label{sub:shrub-depth-interpretability}

We again turn to classical clique-width for an inspiration:
graph classes of bounded clique-width have \MSOi interpretations into the
class of all coloured rooted trees and, in turn, graph classes having an \MSOi
interpretation into those of bounded clique-width still have bounded
clique-width (although the bound on their clique-width is generally much higher).

In one direction, shrub-depth has been defined using
(Definition~\ref{def:tree-model}) a very special form of a simple \MSOi interpretation.
In the other direction, we can go even further than with clique-width itself
(cf.~also Section~\ref{sub:hierarchy}):
the bound on shrub-depth is preserved exactly (and not only asymptotically)
under any \CMSOi interpretations.
In other words, the precise height of a tree is absolutely essential for
\CMSOi interpretability.
The full formal statement follows.

\begin{thm}
\label{thm:shrub-depth-interpretability}
A class $\cf G$ of graphs has a simple \CMSOi interpretation in a class of
finite coloured rooted trees of height at most~$d$, 
if, and only if, $\cf G$~has shrub-depth at most~$d$.
\end{thm}

The `if' direction of Theorem~\ref{thm:shrub-depth-interpretability} follows
immediately from Definition~\ref{def:tree-model}:
for any $m$, the class $\TM d{m}$ has a simple
\MSOi interpretation (or even FO interpretation) in the class
of $m$-coloured tree-models of depth~$d$.
Hence we now give a proof of the `only if' direction of
Theorem~\ref{thm:shrub-depth-interpretability}
consisting of the following sequence of three technical claims.

\begin{lem}[Gajarsk\'y and Hlin\v en\'y \cite{gh15}] 
\label{lem:dtreekern}
There exists a function%
\footnote{Here ${\rm exp}^{(d)}$ stands for the iterated (``tower of height $d$'')
exponential, i.e., ${\rm exp}^{(1)}(x)=2^x$ and ${\rm exp}^{(i+1)}(x)=2^{{\rm exp}^{(i)}(x)}$.}
$R(q,m,d)\leq {\rm exp}^{(d)}\big((q+m)^{\cf O(1)}\big)$
over the positive integers such that the following holds.

Let $T$ be a rooted tree with each vertex assigned one of at most $m$
colours, and let $\phi$ be any \CMSOi sentence with $q$ quantifiers,
such that the least common multiple of the $b$ values of all 
$\!\mod\!_{a,b}$ predicates in $\phi$ equals~$M$.
Take any node $u\in V(T)$ such that the subtree $T_u\subseteq T$ rooted
at~$u$ is of height~$d$, and denote by $U_1,U_2,\dots,U_k$ the connected
components of $T_u-u$ (their roots are thus all the $k$ sons of~$u$).

Assume that there exists a (sufficiently large) 
subset of indices $I\subseteq\{2,\dots,k\}$, where $|I|\geq R(q,m+M,d)+M-1$, 
such that there are colour-preserving isomorphisms from $U_1$ to each $U_i$, $i\in I$.
Choose any $J\subseteq\{1\}\cup I$, $|J|=M$, and take
the subtree $T'=T-\bigcup_{j\in J}V(U_j)$.
Then $T'$ behaves the same with respect to $\phi$ as~$T$, precisely, 
$T\models\phi$ $\iff$ $T'\models\phi$.
\end{lem}

Lemma~\ref{lem:dtreekern} and, in particular, the operation of obtaining 
$T'$ from $T$ as in the lemma,
will be useful in the following generalized setting of a reduction.
Assume that $R':\mathbb N\to\mathbb N$ is an arbitrary non-decreasing
function and $M$ is a positive integer
(for use with Lemma~\ref{lem:dtreekern}, we can have~$R'(i):=R(q,m+M,i)$),
and $T$ is a coloured rooted tree of height~$d$.
Inductively for $i=1,2,\dots,d$, we do the following:
For every $w \in V(T)$ such that $T_w$ is of height $i$, 
consider the components of $T_w - w$ partitioned into equivalence
classes according to the existence of a colour-preserving isomorphism.
In each of these classes whose cardinality is at least $R'(i)+M$,
we repeatedly remove $M$-tuples of components until the cardinality reaches
$R'(i)+c$ where $0\leq c<M$.
Let $T''$ be the resulting ``reduced'' subtree of~$T$. 
In such situation we say that $T$ is {\em$R'$-reduced (modulo~$M$) to~$T''$}.
Observe that $T''$ is of bounded size depending only on 
$R'$, $M$ and $d$, and independent of the size of~$T$.

We continue with the technical claims leading to
Theorem~\ref{thm:shrub-depth-interpretability}. 
Imagine a situation in which we have a graph (tree) automorphism taking a vertex $x_1$
to a vertex $x_2$, and similarly an automorphism taking $y_1$ to~$y_2$.
Then it is generally not true that there would exist an automorphism taking
the pair $(x_1,y_1)$ to the pair $(x_2,y_2)$.
The next lemma establishes a simple additional condition under which 
the previous becomes always true.
We need the notion of an orbit.
The binary relation on the vertex set of a graph
defined as `$x_1\sim x_2$ iff there is an automorphism 
taking $x_1$ to~$x_2$' is an equivalence, and its equivalence classes are
called the {\em vertex automorphism orbits}.

Note that all automorphisms in this section are colour-preserving.
\begin{lem}
\label{lem:automorphpair}
Let $T$ be a coloured rooted tree.
Assume that $X,Y$ are vertex automorphism orbits of $T$,
and $x_1,x_2 \in X$ and $y_1,y_2 \in Y$ are chosen arbitrarily.
Let $z_i=x_i\wedge y_i$, $i=1,2$, denote the least common ancestor 
of $x_i,y_i$ in $T$.
If $dist_T(x_1,z_1)=dist_T(x_2,z_2)$ and
$dist_T(y_1,z_1)=dist_T(y_2,z_2)$, then there is an
automorphism of $T$ taking the pair $(x_1,y_1)$ onto $(x_2,y_2)$.
\end{lem}
\begin{proof}
We carry on the proof by induction on $d=dist_T(x_1,z_1)+dist_T(y_1,z_1)$.

The base case of $d=0$ is trivial (since $x_1=y_1$ and $x_2=y_2$).
Consider now an induction step from $d$ to $d+1$ where $dist_T(x_1,z_1)\geq1$.
Let $x_1',x_2'$ be the parent nodes of $x_1,x_2$, respectively,
and let $X'$ denote the set of parent nodes of all the members of $X$.
Then $X'$ is a vertex orbit of $T$, too.
By inductive assumption, there is an automorphism $\tau$ of $T$ taking 
the pair $(x_1',y_1)$ onto $(x_2',y_2)$.
If $\tau(x_1)=x_3$, then $x_3$ is a child of $x_2'$,
and the subtree of $T$ rooted at $x_3$ is isomorphic to that of $x_2$
by transitivity.
Therefore, we may without loss of generality assume $x_3=x_2$,
and the induction step is complete.
\end{proof}

\begin{figure}[tb]
\begin{center}
\begin{tikzpicture}[yscale=0.75]
\tikzstyle{every node}=[draw, shape=circle, inner sep=1.2pt, fill=black]
\draw[fill=lightgray] (3,4) -- (1,1) -- (3,1) -- (3,2) -- (4,2) -- (3,4) ;
\draw[fill=lightgray] (1,1) -- (0.5,0) node[fill=white] {} -- (1.2,0) -- (1,1) ;
\draw[fill=lightgray] (3,1) -- (2.5,0) -- (3.2,0) -- (3,1) ;
\draw[fill=lightgray] (4,2) -- (3.5,0) node[fill=white] {} -- (4.5,0) -- (4,2) ;
\draw[thick] (1,1) node[fill=white] {} -- (1.5,0) node {} ;
\draw[thick] (3,1) node[fill=white] {} -- (2.1,0) node {} ;
\draw[thick] (4,2) node[fill=white] {} -- (4.5,1) node[fill=lightgray] {} -- (5,0) node {} ;
\end{tikzpicture}
\caption{An illustration of the operation of growing new leaves (black dots)
	from selected original nodes (white dots) of the depicted rooted tree of height~$d$.
	If a selected original node already is at a distance $d$ from the root
	(the bottom layer), then no new leaf is actually grown.}
\label{fig:growing-leaves}
\end{center}
\end{figure}
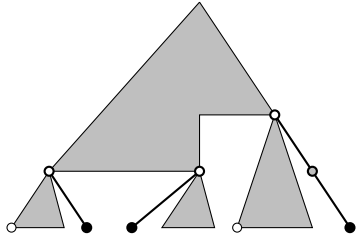

\begin{lem}
\label{lem:tm-growing-leaves}
Assume that a class $\cf G$ of graphs has a simple \CMSOi interpretation
$I$ in a class $\cf T_d$ of finite coloured rooted trees of height at most~$d$.
Then there exists $m$ such that the following holds:
every graph $G\in\cf G$, where $G=I(T)$ for some $T\in\cf T_d$,
has an $m$-coloured tree-model $U$ of depth~$d$.

Furthermore, the rooted tree $U$ is obtained from $T$ by ``growing leaves'' from
those nodes of~$T$ that belong to the domain of~$I$ and have distance less
than~$d$ from the root.
Specially, if the domain of $I$ is a subset of the leaves of $T$
and all leaves of $T$ are at a distance $d$ from the root, then~$U\subseteq T$.
\end{lem}
Here the operation of {\em growing a leaf\/} from a node $u$
of a rooted tree $T$ of height $d$ means to add a new branch (a path) from
$u$ to a (new) leaf $u'$ such that the distance from the root to $u'$ is exactly $d$.
Only one new leaf is grown from a node~$u$, and only when
the distance of $u$ from the root is less than~$d$ (otherwise~$u'=u$).
See Figure~\ref{fig:growing-leaves}.

\begin{proof}
Let the simple \CMSOi interpretation $I(\cf T_d)=\cf G$
be given by the formulas $I=(\alpha,\beta)$.
Recalling the definition of a simple interpretation,
every $G\in\cf G$ is interpreted in some coloured tree $T_G\in\cf T_d$ as follows:
$V(G)=\{x\in V(T_G): T_G\models\alpha(x)\}$ and
$E(G)=\{xy: x,y\in V(G)\wedge T_G\models\beta(x,y)\}$. 
To assist readers' understanding, we remark that we can evaluate the domain
$V(G)\subseteq V(T_G)$ at the beginning, and so we will not 
deal with $\alpha$ any more.

For the purpose of directly applying Lemma~\ref{lem:dtreekern},
we transform $\beta$ into a closed sentence, 
$\beta'\equiv\exists x,y\big(L(x)\wedge L(y)\wedge \beta(x,y)\big)$,
where $L$ is a new vertex label (which will be later added to existing 
colours of specified nodes).
Since $\beta'$ is a finite formula, it can ``see'' at most some 
$m'$ of the colours of $\cf T_d$ including new~$L$.

Let $G \in \cf G$ be a fixed graph and let $T=T_G$, as above. 
Let $q$ be the number of quantifiers in $\beta'$, and let $M$ be the least
common multiple of the $b$ values of all $\!\mod\!_{a,b}$ predicates occurring in~$\beta'$.
For the function $R$ from Lemma~\ref{lem:dtreekern}, let $R'(i) := R(q,m'+M,i)+M + 2$.  
Choose an arbitrary pair~$u,v\in V(G)$.
We may straightforwardly get a coloured subtree $T_0^{u,v}\subseteq T$
(not unique) such that the tree $T$ is 
$R'$-reduced (modulo~$M$) to $T_0^{u,v}$ and $u,v\in V(T_0^{u,v})$.
If we add the label $L$ precisely to $u$ and~$v$,
we shortly denote the resulting structures by $T[L(u),L(v)]$ and
$T_0^{u,v}[L(u),L(v)]$.

\medskip
The first crucial step of the proof is to observe that
$T[L(u),L(v)]\models\beta'$ $\iff$ $T_0^{u,v}[L(u),L(v)]\models\beta'$.
We can easily show this from Lemma~\ref{lem:dtreekern} along the inductive 
definition of $T$ being $R'$-reduced (modulo~$M$) to $T_0^{u,v}$.
Following the notation of Lemma~\ref{lem:dtreekern}, we assume
an intermediate step $T'$ of the reduction process, and a node $w$ of $T'$
such that there is an isomorphism class $\cf I$ of the components of
$T'_w[L(u),L(v)]-w$ of size at least $R'(i)-2\geq R(q,m'+M,i)+M$.
(Here `$-2$' accounts for the possibility that, in the corresponding
isomorphism class of $T'_w-w$, up to two of the components got the label $L$,
and so they are not a part of $\cf I$ in $T'_w[L(u),L(v)]-w$.)
Let $T''\subseteq T'$ result by removing $M$ components of~$\cf I$.
Then, by Lemma~\ref{lem:dtreekern}, we have
$T'[L(u),L(v)]\models\beta'$ $\iff$ $T''[L(u),L(v)]\models\beta'$
and we finish by induction.

Second, one may observe that for any other pair~$u',v'\in V(G)$,
the tree $T_0^{u',v'}\subseteq T$ is always isomorphic to $T_0^{u,v}$.
Hence we can choose one universal representative, say $T_0:=T_0^{u,v}$.
Note that $T_0$ is of bounded size depending only $q, m'$ and $d$, and
independent of the size of~$T$.
Consequently, for an arbitrary pair $u',v'\in V(G)$, we can
determine whether or not $u'v'$ forms an edge in $G$ by testing if~$T_0$
with the ``right assignment'' of $L$ satisfies $\beta'$.
Here the words ``right assignment of $L$ in $T_0$'' implicitly refer
to the images of $u',v'$ under an isomorphism of $T_0^{u',v'}$ to~$T_0$.
Now it only remains to say how to determine these images in $T_0$
within a sought tree-model $U$ of~$G$ (which is, though, nontrivial).

\medskip
From now on, consider the fixed representative~$T_0$.
For $w\in V(T)$, denote by $h(w)$ the distance from $w$ to the root of~$T$.
For any $v \in V(G)$, choose an arbitrary $u\in V(G)$, denote by $\cf O$ the
vertex automorphism orbit of $v$ in $T_0^{u,v}$ and by
$Or(v)$ the image of $\cf O$ under an isomorphism of $T_0^{u,v}$ to~$T_0$.
Observe that $Or(v)$, as an automorphism orbit, does not depend on our choice
of $u$ and of an isomorphism of $T_0^{u,v}$ to~$T_0$.
Let $S$ be the rooted Steiner tree of $V(G)$ in $T=T_G$,
which is the minimal rooted subtree $S\subseteq T$ containing~$V(G)$.
From $S$ we construct the (uncoloured) tree~$U$ by growing leaves from 
all nodes $w$ of $S$ such that $w\in V(G)$ and~$h(w)<d$, 
in order to literally satisfy Definition~\ref{def:tree-model}.
Each such newly grown leaf $w'$ will now interpret the corresponding vertex 
of~$G$ instead of the original~$w$, that is, we will identify $w\in V(G)$ with
$w'\in V(U)$, and set $h'(w'):=d-h(w)$ and $Or(w'):=Or(w)$.
For $v\in V(G)$, such that no new leaf has been grown from $v$,
we have~$h'(v)=0$.

In the third (and last) step of the proof, we show that $U$ together with 
information carried by $T_0$ and $Or$, $h'$ are sufficient 
to decide whether a pair $u,v\in V(G)$ forms an edge of~$G$.
More precisely, we have $u,v\in V(T)\cap V(G)$ and $u',v'$ are the corresponding
grown leaves in $U$, or $u'=u$ ($v'=v$) if no leaf has been grown from
$u$ ($v$) in~$U$.
Let $z=u\wedge v$ in~$T$ (cf.\ Lemma~\ref{lem:automorphpair}).
We simply determine $dist_T(u,z)=dist_U(u',z)-h'(u')$
and analogously for~$v$.
Then, knowing $dist_T(u,z)$, $dist_T(v,z)$ and $Or(u)=Or(u')$, $Or(v)=Or(v')$,
by Lemma~\ref{lem:automorphpair}, determines uniquely up to an automorphism
a pair $u_0,v_0\in V(T_0)$ such that $(u,v)$ maps to $(u_0,v_0)$ under an
isomorphism of $T_0^{u,v}$ to~$T_0$.
Summarizing all the arguments, we get that
$T_0[L(u_0),L(v_0)]\models\beta'$ $\iff$
 $T_0^{u,v}[L(u),L(v)]\models\beta'$ $\iff$
 $T[L(u),L(v)]\models\beta'$ $\iff$ $T\models\beta(u,v)$,
as desired.

Therefore, we can build a tree-model $U$ of $G$
(Definition~\ref{def:tree-model}) by assigning
the colour $\langle Or(v),h'(v)\rangle$ to each leaf $v$ of~$U$
and giving it the signature determined from $T_0$ by the latter argument.
\end{proof}

\subsection{Stability under transductions}

The first important consequence of
Theorem~\ref{thm:shrub-depth-interpretability}
is that the shrub-depth of a graph class is preserved under non-copying
\CMSOi transductions.

\begin{thm}
\label{thm:nc-transduction}
Let $d\geq1$ be an integer, $\cf G$ be a graph class of shrub-depth~$d$,
and $\tau$ be a non-copying \CMSOi transduction.
Then the shrub-depth of the transduction image $\tau(\cf G)$ 
is at most~$d$.
\end{thm}

\begin{proof}
Let $\cf G\subseteq \TM dm$,
and let $I_1$ denote the corresponding interpretation of $\cf G$ in 
a class of $m$-coloured tree-models of depth $d$.
Assume $\tau=\tau_0\circ\varepsilon$ where $\tau_0$ is a basic
transduction and $\varepsilon$ is an expansion by $p$ unary predicates.
Since each of the $p$ predicates can be encoded by a binary label added to
the above $m$ colours, we have got that
$\varepsilon(\cf G)$ has an interpretation $I_1'$ in a class $\cf T$ of
$(2^pm)$-coloured rooted trees of height $d$.
Let $I_0$ be the simple \CMSOi interpretation underlying~$\tau_0$.
Then $\tau(\cf G)\subseteq I_0\big(I_1'(\cf T)\big)$.
Since $I_0\circ I_1'$ is again a \CMSOi interpretation,
the latter class has shrub-depth at most $d$ by
Theorem~\ref{thm:shrub-depth-interpretability} and the claim follows.
\end{proof}

We now look at the more general case of copying \CMSOi transductions.
One cannot immediately extend Theorem~\ref{thm:nc-transduction}
towards this case since, for example,
a $2$-copying transduction of the class of edge-less graphs
(shrub-depth~$1$) contains all perfect matchings (shrub-depth~$2$).
This is, however, only a technical problem which we resolve simply by
allowing ``copying'' tree-models here.

Informally, a {\em $k$-copied tree-model\/} is a tree-model $T$
as in Definition~\ref{def:tree-model}, with an exception that
every leaf of $T$ holds an ordered $\leq k$-tuple of distinct vertices of~$G$
and the existence of an edge can depend also on the tuple of a vertex
and its index within the tuple.
This is formally stated (with a twist) as follows:

\begin{defi}[$k$-copied tree-model]
\label{def:k-tree-model}
A graph $G$ has a {\em $k$-copied tree-model of $m$ colours and ``depth'' $d$}
if $G$ has an ordinary tree-model $T$ of $m$ colours and depth $d+1$ such
that every node of $T$ at distance $d$ from the root has at most~$k$
descendants (the leaves).
The class of all graphs $G$ having such a $k$-copied tree-model
is denoted by $\TMC dmk$.

A class of graphs $\cf G$ has {\em copying shrub-depth} $d$
if there exist $m,k$ such that $\cf G\subseteq\TMC dmk$,
while for all natural $m',k'$ it is $\cf G\not\subseteq\TMC{d-1}{m'}{k'}$.
\end{defi}

Notice that $\TMC dm1=\TM dm$, but this is not true in general for
higher values of~$k$.
It is not too difficult to observe that every graph class of copying shrub-depth~$d$ is
contained in a suitable $k$-copying transduction of a class of ordinary
shrub-depth~$d$.
We complement this observation with the following:

\begin{thm}
\label{thm:kc-transduction}
Let $d\geq1$, $\cf G$ be a graph class of copying shrub-depth~$d$,
and $\tau$ be a \CMSOi transduction.
Then the copying shrub-depth of $\tau(\cf G)$ is again at most~$d$.
\end{thm}

\begin{proof}
Let $\tau=\tau_0\circ\gamma\circ\varepsilon$ where 
$\tau_0$ is a basic \CMSOi transduction, $\gamma$ is a $k$-copy operation
and $\varepsilon$ is an expansion by $p$ unary predicates.

We remark that, thanks to transitivity of transductions,
it is enough to prove this statement in the case that $\cf G$ is
of ordinary shrub-depth~$d$.
So, as in the proof of Theorem~\ref{thm:nc-transduction},
$\cf G\subseteq \TM d{m_1}$,
and let $I_1$ denote an \MSOi interpretation of $\cf G$ in 
a suitable class $\cf U$ of $m_1$-coloured tree-models of depth $d$.
Then, again as before, we may say that
$\varepsilon(\cf G)$ has an interpretation $I_1'$ in the corresponding 
class $\cf U'$ of $(2^pm_1)$-coloured tree-models of depth~$d$, that is,
$\varepsilon(\cf G)=I_1'(\cf U')$.

In the next step, we aim to show that the class $\cf K=\gamma(I_1'(\cf U'))$ 
of relational structures (see Definition~\ref{def:transduction} for~$\gamma$)
actually has a simple interpretation in a suitable class $U^+$ of trees.
Here it is important that the domain of $I_1'$ (which is to be copied by~$\gamma$) 
is restricted to leaves of the trees of~$\cf U'$.
For $U\in\cf U'$, let $U^+$ be the $(k2^pm_1)$-coloured tree-model of  
depth~$d+1$ constructed as follows:
for each leaf $u$ of $U$ of colour $c$, add $k$ new descendant leaves with
the parent $u$ and of distinct colours $(c,1),\ldots,(c,k)$.
Actually, $U^+$ is also a $k$-copied tree-model of ``depth''~$d$ 
according to Definition~\ref{def:k-tree-model}.
Let $\cf U^+=\{U^+:U\in\cf U'\}$.

From the definition of $\gamma$, one can easily come up with a simple \MSOi
interpretation $I_2=(\nu_2,\mu_2)$ defining in $\cf U^+$ the underlying
graphs of the structures of~$\cf K$, and an \MSOi formula $\sigma_2$
defining the binary relation $\sim$ of $\gamma$ (while the unary relations
$P_i$ of $\gamma$ are already encoded in the colours of $\cf U^+$).
We hence give a simple interpretation $I^+=(\nu_2,\mu_2,\sigma_2)$
(naturally extending Definition~\ref{def:interpretation} for~$I^+$) such that
$I^+(\cf U^+)$ equals~$\cf K$.

Finally, let $I_0$ be the simple \CMSOi graph interpretation underlying~$\tau_0$.
We summarize that $\tau(\cf G)=\tau_0(\cf K)$ is contained in 
$(I_0\circ I^+)(\cf U^+)$.
By Lemma~\ref{lem:tm-growing-leaves}, there then exists $m$ such that
every graph $H$, where $H=(I_0\circ I^+)(U_1)$ for some $U_1\in\cf U^+$,
has an $m$-coloured tree-model $U_2$ of depth~$d+1$.
Moreover, we have $U_2\subseteq U_1$ since the domain of $I_0\circ I^+$ 
is restricted to the leaves of~$U_1$ (recall~$I_1'$),
and so $U_2$ is also a $k$-copied tree-model of
$m$ colours and ``depth''~$d$ of the graph~$H$.
Consequently, $\tau(\cf G)$ has copying shrub-depth at most~$d$.
\end{proof}

\subsection{On \MSOi-transduction hierarchy}
\label{sub:hierarchy}

The second interesting consequence of Theorems
\ref{thm:shrub-depth-interpretability} and~\ref{thm:kc-transduction}
claims that every graph class of bounded shrub-depth ``falls under'' 
precisely one of the integer values of copying
shrub-depth according to transduction equivalence (both \MSOi and \CMSOi).
This coming result is close to the main result of Blumensath and
Courcelle \cite[Theorem~6.4]{bc10} completely characterizing
the related \MSOii-transduction hierarchy (precisely, the \MSO transduction hierarchy
of the vertex-edge incidence structures of undirected graphs).

We begin with some necessary technical terms.
Fix a logical language of transductions (such as \MSOi or \CMSOi of
simple undirected graphs).
For two classes of relational structures (graphs in our case)
$\cf K,\cf L$, we write $\cf K\sqsubseteq\cf L$ if there exists a
transduction $\tau$ such that $\cf K\subseteq\tau(\cf L)$.
Similarly we write $\cf K\sqsubsetneq\cf L$ if 
$\cf K\sqsubseteq\cf L$ but $\cf L\not\sqsubseteq\cf K$,
and $\cf K\equiv\cf L$ if both $\cf K\sqsubseteq\cf L$ and
$\cf L\sqsubseteq\cf K$ hold true.

The relation $\sqsubseteq$ forms a quasi-ordering on the considered 
classes of structures, as can be easily seen~\cite{bc10}.
The research question here is to describe the underlying ordering of
$\sqsubseteq$, i.e., the {\em transduction hierarchy}.
Unlike for the aforementioned completely solved case of \MSOii-transduction hierarchy,
only a weaker partial outcome has been known regarding \MSOi transductions:
\begin{thm}[Blumensath and Courcelle \cite{bc10}]
\label{thm:hierarchy0}
Let $\cf T_d$ denote the class of all finite rooted trees of height at most~$d$.
In the scope of either \MSOi or \CMSOi transductions, the following holds
$$\cf T_1 \sqsubsetneq \cf T_2 \sqsubsetneq \cf T_3 \sqsubsetneq
	\ldots \sqsubsetneq \cf T_d \ldots \,.$$
\end{thm}

Here we provide a full solution of \MSOi-transduction hierarchy
for graph classes of bounded shrub depth (thus completing and extending
Theorem~\ref{thm:hierarchy0}):

\begin{thm}
\label{thm:hierarchy}
For any graph class $\cf G$ of bounded shrub-depth there is 
an integer $d$ such that $\cf G\equiv\cf T_d$ in the scope of either 
\MSOi or \CMSOi transductions.
\end{thm}

Before getting to the proof, we first establish two supplementary lemmas.
Let $\overline T_d^{\,r}$ denote the complete rooted $r$-ary tree of height~$d$.

\begin{lem}
\label{lem:depth-down}
Let $\cf G$ be a graph class of bounded shrub-depth.
If there exist integers $d,m,r$ such that every graph $G\in\cf G$
has an $m$-coloured tree-model of depth $d$ not containing $\overline T_{d}^{\,r}$ 
as a rooted subtree, then the copying shrub-depth of $\cf G$ is at most~$d-1$.
\end{lem}

\begin{proof}
Let $G$ have a tree-model with the underlying rooted tree $U$ of height
$d$ such that $\overline T_{d}^{\,r}\not\subseteq U$.
We have borrowed the following high-level proof idea from \cite[Lemma~4.12]{bc10}.

Let $R\subseteq V(U)$ be the minimal (by inclusion) set of nodes
such that $R$ contains all the leaves of $U$, 
and $R$ contains every internal node of $U$ which has at least $r$ of its children in~$R$.
Let $F\subseteq E(U)$ be the set of edges having the child end in $R$ and the
parent end in $V(U)\setminus R$.
The root of $U$ is not in~$R$ since $\overline T_{d}^{\,r}\not\subseteq U$.
So, every root-to-leaf path in $U$ contains an edge from~$F$.
Moreover, every internal node of $U$ that is not in $R$, has at most $r-1$
of its child edges in~$F$ (or it would be added to~$R$).

Now, to every non-leaf edge $f\in F$ with the parent end $v$ we assign a label
$\ell_f=(i,j)$, where $0\leq i\leq d-2$ is the distance of $v$ from the root and 
$1\leq j<r$ is the index of $f$ among all $F$-edges incident with~$v$
(in an arbitrary fixed ordering of the children).
Then, in the subtree $U_f$ below~$f$ in~$U$,
we subdivide all the leaf edges of $U_f$ and
we add the label $\ell_f$ to (the colours of) the leaves of $U_f$.
Then we contract $f$.
Let $U'$ denote the resulting labelled tree (which is again of height~$d$).
One can routinely verify that information additionally provided by the added
labels ($\ell_f$) is sufficient for $U'$ to be a tree-model of~$G$, too.
Furthermore, our construction of~$U'$ guarantees that $U'$ actually is
an $(r-1)$-copied tree-model of depth~$d-1$, as in Definition~\ref{def:k-tree-model}.
Since this holds, with the same $d,r$, for every $G\in\cf G$,
the copying shrub-depth of $\cf G$ is at most~$d-1$.
\end{proof}

\begin{lem}
\label{lem:get-tree}
For every integers $d,m\geq1$ there exists a non-copying 
\MSOi transduction $\sigma_{d,m}$ such that the following holds:
if, for an integer $r$ and a graph $G\in\TM dm$,
every $m$-coloured tree-model of depth $d$ of $G$ contains
the tree $\overline T_{d}^{\,r}$ as a rooted subtree,
then $\overline T_{d}^{\,r-1}\in \sigma_{d,m}(G)$.
\end{lem}

We remark that Blumensath and Courcelle [personal communication]
have established a statement similar to Lemma~\ref{lem:get-tree}, 
but it has not been published. 
For the sake of completeness, we give our independent proof here.

\begin{proof}
Our strategy is to construct a very specific tree-model $U$ of~$G$,
such that we can interpret in suitably labelled $G$ 
a tree $U'\subseteq U$ which is ``nearly~$U$'' in the sense that 
only one child of each node of the underlying tree of $U$ is missing 
(it is used to represent this node instead).
From the assumption $\overline T_{d}^{\,r}\subseteq U$ we can then conclude
that $\overline T_{d}^{\,r-1}$ will be contained in
the respective non-copying transduction image $\sigma_{d,m}(G)$ of this
interpretation.

We use technical terms from \cite{gh15}.
Assume $T$ is a tree-model of~$G$, with an internal node $u$, and 
let $W$ be the set of leaves of $T_u$.
We say that a tree-model $T'$ is obtained from $T$ by 
{splitting $T_u$ along} $X\subseteq W$ if
a disjoint copy $T_u'$ of $T_u$ with the same parent is added into $T$, 
and then $T_u$ is restricted to the leaves in $W\sem X$
while $T_u'$ is restricted to those in the corresponding copy of~$X$.
A tree-model $T$ is {\em unsplittable} if no such splitting $T'$ of $T$
represents the same graph $G$ as $T$ does.

Fix now an unsplittable tree-model $U$ of $G$ (which obviously exists, by
repeated splits).
Let $U$ be $Q$-reduced to $U_0\subseteq U$ for $Q\equiv2$
(cf. Section~\ref{sub:shrub-depth-interpretability} for ``reduced''),
where $U_0$ is of constant size depending on $d,m$.
We colour the vertices of $G$ by their colours in $T$,
and additionally give individual distinguishing labels to those 
(constantly many) vertices which are the leaves of~$U_0$.
\begin{itemize}
\item By \cite[Lemma~5.10]{gh15},
there exists an FO-definable relation $\sim$ (depending on~$U_0$)
on the vertices not in $U_0$
such that $G\models x\sim y$ if, and only if,
$x,y$ are leaves of the same component (subtree) of $U-V(U_0)$.
\end{itemize}

From $\sim$ one can recursively construct FO-definable relations $\approx_i$
on $V(G)$, for $i=1,\dots,d$, such that the following holds:
$G\models x\approx_i y$ if, and only if, there is a node $w$ of $U$
such that $U_w$ is of height $i$ and $x,y\in V(U_w)$.
The precise technical details are analogous to the proof of
\cite[Theorem~5.14]{gh15}, and we refrain from repeating them here.

Finally, from each equivalence class of $\approx_1$ we choose an arbitrary
representative, and give all these representatives a new label $\nu_1$.
Recursively, from each equivalence class of $\approx_i$, $i\geq2$, 
we choose a representative among those labelled $\nu_{i-1}$,
and give them an additional label $\nu_i$.
We can now easily interpret the desired tree $U'$ in $G$ using the
relations $\approx_i$ and the labels $\nu_i$.
Consequently, since $\overline T_{d}^{\,r-1}\subseteq U'$,
this gives $\sigma_{d,m}$ such that 
$\overline T_{d}^{\,r-1}\in\sigma_{d,m}(G)$. 
\end{proof}

We can now finish the main result.

\begin{proof}[Proof of Theorem~\ref{thm:hierarchy}]
Let us consider any graph class $\cf G$ of copying shrub-depth~$d$,
and fix $m$ be such that $\cf G\subseteq\TM dm$.
Since the copying shrub-depth of $\cf G$ is not $d-1$, by
Lemma~\ref{lem:depth-down},
we obtain that for every integer $r$ there exists $G_r\in\cf G$
such that, every $m$-coloured tree-model of $G_r$ of depth $d$
contains $\overline T_{d}^{\,r}$ as a rooted subtree.
Then, by Lemma~\ref{lem:get-tree}, there is an \MSOi transduction
$\sigma_{d,m}$ (independent of~$r$) such that
$\overline T_{d}^{\,r-1}\in \sigma_{d,m}(G_r)$.
Hence, $\overline{\cf T}_d:=\big\{
 \overline T_{d}^{\,s}: s\in\mathbb N \big\}
 \subseteq\sigma_{d,m}(\cf G)$.
Since $\cf T_d$ is easily a transduction of $\overline{\cf T}_d$,
we conclude that $\cf G\equiv\cf T_d$.
\end{proof}

\section{Concluding notes}
\label{sec:conclu}

The structural properties of classes of bounded shrub-depth,
in Section~\ref{sec:shrub}, leave one important question widely open:
what is a nice asymptotic structural characterization of graph 
classes of unbounded shrub-depth?
There are indications, related to matroid theory and to the notion
of rank-depth by DeVos, Kwon and Oum, that the following might be 
the ultimate answer:
\begin{itemize}
\item[] \cite[\bf Conjecture~6.3]{hkoo16}
A class $\cf C$ of graphs is of bounded shrub-depth if, and only if,
there exists an integer $t$ such that no graph $G\in\cf C$ 
contains a path of length $t$ as a vertex-minor.
\end{itemize}

On the other hand, in relation to the transduction hierarchy studied
in Section~\ref{sec:transdhier}, the following seems a plausible conjecture:
\begin{conj}
A class $\cf G$ of graphs is of bounded shrub-depth if, and only if,
for every \CMSOi transduction $\tau$ there exists an integer $t$ such that
$P_t\not\in\tau(\cf G)$.
\end{conj}
While the `only if' direction follows from Theorems~\ref{thm:Path-TM}
and~\ref{thm:kc-transduction}, the `if' direction can be seen as a weaker
form of \cite[Conjecture~6.3]{hkoo16} since a vertex-minor can be captured
by a non-copying \CMSOi transduction.

\smallskip
Finally, we briefly mention a natural extension of the shrub-depth
notion to general relational structures (e.g., digraphs).
Regarding Definition~\ref{def:tree-model} of a tree-model,
the extension is straightforward.
For any (finite) signature of a relational structure~$\cf S$,
the domain of $\cf S$ is again the set of leaves of~$T$,
and we consider (one of) its $k$-ary relational symbols $R$.
We state that, for an ordered $k$-tuple $x_1,\dots,x_k$ from the domain,
$R(x_1,\dots,x_k)$ depends only on the colours of $x_1,\dots,x_k$
and the shape of the rooted Steiner tree of the leaves $x_1,\dots,x_k$.
Hence we can define the shrub-depth as in Definition~\ref{def:shrub-depth}
for any class of relational structures of a given finite signature.
(Notice, though, that SC-depth does not extend this way.)

With the previous definition, one may readily extend
Theorem~\ref{thm:shrub-depth-interpretability} to classes of relational
structures of a fixed finite signature.
In fact, it is enough to provide a corresponding extension of
the technical Lemma~\ref{lem:automorphpair},
and the rest of the arguments go through smoothly.
The question of the lower levels of the \MSO-transduction hierarchy,
as in Theorem~\ref{thm:hierarchy}, of such classes is left for future
investigation.

\bibliographystyle{abbrv}
\bibliography{gtbib}

\end{document}